\documentclass[10pt, conference, compsocconf]{IEEEtran}
\usepackage{amsmath, amssymb}
\usepackage{graphicx}
\usepackage{subfigure}
\usepackage{verbatim}
\usepackage{color}
\usepackage{xspace}
\usepackage{caption}% http://ctan.org/pkg/caption
\usepackage{bibspacing}
\usepackage{authblk}
\newtheorem{definition}{Definition}
\newtheorem{theorem}{Theorem}

\newtheorem{proposition}{Proposition}
\newtheorem{example}{Example}

\newcommand{\myhrule}{\rule[.5pt]{\hsize}{.5pt}}

\newcommand{\eat}[1]{}
\newcommand{\stab}{\rule{0pt}{8pt}\\[-1.6ex]}

\newcommand{\sstab}{\rule{0pt}{8pt}\\[-2.4ex]}

\newcommand{\bi}{\begin{itemize}}
\newcommand{\ei}{\end{itemize}}

\newcommand{\mat}[2]{{\begin{tabbing}\hspace{#1}\=\+\kill #2\end{tabbing}}}

\newcommand{\be}{\begin{enumerate}}
\newcommand{\ee}{\end{enumerate}}
\newcommand{\beqn}{\begin{eqnarray*}}
\newcommand{\eeqn}{\end{eqnarray*}}

\newcommand{\stitle}[1]{\vspace{1.5ex}\noindent{\bf #1}}
\newcommand{\etitle}[1]{\vspace{0.8ex}\noindent{\em #1}}

\newcommand{\ie}{\emph{i.e.,}\xspace}
\newcommand{\eg}{\emph{e.g.,}\xspace}
\newcommand{\wrt}{\emph{w.r.t.}\xspace}

\newcommand{\If}{\mbox{\bf if}\ }

\newcommand{\Then}{\mbox{\bf then}\ }

\newcommand{\Do}{\mbox{\bf do}\ }

\newcommand{\For}{\mbox{\bf for}\ }
\newcommand{\Each}{\mbox{\bf each}\ }

\newcommand{\Return}{\mbox{\bf return}\ }

\newcommand{\kw}[1]{{\ensuremath {\mathsf{#1}}}\xspace}

\newcounter{ccc}
\newcommand{\bcc}{\setcounter{ccc}{1}\theccc.}
\newcommand{\icc}{\addtocounter{ccc}{1}\theccc.}

\newcommand{\NP}{{\sc np}\xspace}
\newcommand{\APX}{{\sc apx}\xspace}
\newcommand{\apx}{{\sc apx}\xspace}

\newcommand{\PTIME}{{\sc ptime}\xspace}

\newcommand{\DAG}{\kw{DAG} \xspace}
\renewcommand{\dag}{\kw{DAG} \xspace}%{{\sc dag}}

\renewcommand{\texttt}[1]{{\small\textsf{#1}}}

\newcommand{\len}{\texttt{len}}
\definecolor{gray}{rgb}{0.5,0.5,0.5}

\newcommand{\AFPR}{\kw{AFP}-\kw{reduction}}

\newcommand{\dist}{\texttt{dist}}

\begin{document}

\title{\vspace{-1ex}Inferring the Underlying Structure of Information Cascades\vspace{-1ex}} % the titile need to be changed.

\author{Bo Zong}
\author{Yinghui Wu}
\author{Ambuj K. Singh}
\author{Xifeng Yan}
\affil{Department of Computer Science \\
University of California at Santa Barbara \\
Santa Barbara, CA 93106-5110, USA \\
\{bzong, yinghui, ambuj, xyan\}@cs.ucsb.edu}

\maketitle

\vspace{-2ex}

\begin{abstract}
In social networks, information and influence diffuse among users as cascades.
%The process is typically known as a cascade.
While the importance of studying cascades has been recognized
in various applications, it is difficult to
observe the complete structure of cascades in practice.
Moreover, much less is known on how to infer cascades
based on partial observations.
In this paper we study the cascade inference problem
following the independent cascade model,
and provide a full treatment from complexity to algorithms:
(a) We propose the idea of consistent trees as the inferred structures for
cascades; these trees connect source nodes and observed nodes with
paths satisfying the constraints from the observed temporal
information.
(b) We introduce metrics to measure the likelihood of
consistent trees as inferred cascades, as well as
several optimization problems for finding them.
(c) We show that the decision problems for
consistent trees are in general \NP-complete,
and that the optimization problems are hard to approximate.
(d) We provide approximation algorithms
with {\em performance guarantees} on the quality of the
inferred cascades, as well as heuristics.
We experimentally verify the efficiency and
effectiveness of our inference algorithms, using real and synthetic data.
\end{abstract}


\section{Introduction}
\label{sec:intro}

In various real-life networks,
%\eg social networks, recommendation networks,
%collaboration networks,
users frequently
exchange information and influence each other. The information (\eg messages, articles,
recommendation links) is typically created from a user and spreads via links among users, leaving
a trace of its propagation.
%For a single piece of information,
Such traces are typically represented as trees, namely, {\em information cascades},
where (a) each node in a cascade is associated with the time step at which it receives the information,
%and sends the information to its neighbors,
and (b) an edge from a node to another indicates that a user propagates the information to and %thus,
{\em influences} its neighbor~\cite{bikhchandani1992theory,GOLDENBERG}.
%Information cascade is {\em gives an intuition here}.

A comprehensive understanding and analysis of cascades benefit various emerging
applications in social networks~\cite{ChaBAG12, KEMPE_KDD}, viral marketing~\cite{arthur2009pricing, Domingos:2001a:KDD01,Richardson:2002A:KDD02}, and
recommendation networks~\cite{leskovec2006patterns}. %, among others.
%The information transmission behavior often follows a specified {\em cascade model}~\cite{dave2011modelling,Saito:2008a:KES08}.
In order to model the propagation of information, various
{\em cascade} models have been developed~\cite{dave2011modelling,Song:2007a:WWW,2011:Wang:CoRR}.  % more to be added.
Among the most widely used models is the \emph{independent cascade model}~\cite{KEMPE_KDD},
where each node has only one chance to influence its inactive neighbors, and
each node is influenced by at most one of its neighbors independently. %(in an arbitrary order)
 % {\em gives an intuition here}.
%
%
Nevertheless, it is typically
difficult to observe the entire cascade in practice, due to the noisy graphs with missing data, or
data privacy policies~\cite{kossinets2006effects, Sadikov:2011a:WSDM11}. It is important to
develop techniques that can {\em infer} the cascades using partial information.
%only
%the information from the observed part of the cascades. %The need is evident in {\em add application areas}.
Consider the following example.

%%%%%%%%%%%%%%%
\begin{figure}[tb!]
%\vspace{-3ex}
\centerline{\includegraphics[scale=0.5]{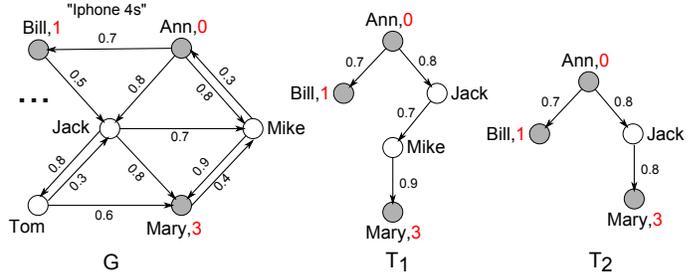}}
%\vspace{-2.5ex}
\caption{A cascade of an Ad (partially observed) in a social network $G$ from user \texttt{Ann}, and
its two possible tree representations $T_1$ and $T_2$.}
\label{fig-exa-cascade}
\vspace{-3ex}
\end{figure}
%%%%%%%%%%%%%%%%

\eat{
In this paper, we study the {\color{blue} cascade prediction problem (temporary name)}
on \emph{independent cascade model}~\cite{}.
The reason why we focus on these two models is two-fold.
First, both of them are widely studied probabilistic models on
cascading behavior~\cite{}; second, it is easier to proceed the discussion,
since cascades based on these two models unfold with tree-like structures,
and diffuse in a progressive way where nodes being active cannot be switched
back to being inactive {\color{red} (currently, I do not know if we can lift this assumption)}.
}

% motivation
%%%%%%%%%%%%%%%%%
\begin{example}
\label{exa-cascade}
The graph $G$ in Fig.~\ref{fig-exa-cascade} depicts a fraction of a social network (\eg Twitter),
where each node is a user, and each edge represents an information exchange.
For example, edge $(\texttt{Ann}, \texttt{Bill})$ with a weight $0.7$
represents that a user \texttt{Ann} sends an advertisement (Ad) about a released product (\eg ``Iphone 4s'')  with probability $0.7$.
%if a user \texttt{Ann} tweets an advertisement(Ad) about a released product (\eg ``Iphone 4s''), another user \texttt{Bill} would retweet the Ad with probability $0.7$.
 To identify the impact of an Ad strategy, a company would like to
know the complete cascade starting from their agent \texttt{Ann}.
Due to %the limited data resource
data privacy policies, the observed information %includes several nodes and
%the time they are influenced:
may be limited:
(a) at time step $0$, \texttt{Ann} posts an Ad about ``Iphone 4s'';
%(b) at time step 1, Bill is influenced by Ann and retweets the Ads; (c) by time step 3, the cascade reaches Mary and Mary retweets it.
(b) \textbf{at} time step $1$, \texttt{Bill} is \textbf{influenced} by \texttt{Ann} and retweets the Ad;
(c) \textbf{by} time step $3$, the Ad reaches \texttt{Mary}, and \texttt{Mary} retweets it.
As seen,  %Moreover, % according to their past observation,
the information diffuses from one user to his or her neighbors with different probabilities, %each user accepts the information from one of his or her friends,
represented by the weighted edges in $G$. %and sends the information to his or her friends with different probabilities, represented by
%the labeled edges in $G$. %
%In other words,
Note that the cascade unfolds as a \textbf{tree},
rooted at the node \texttt{Ann}. % in this case.

%in other words, the cascades are trees.

To capture the entire topological information of the cascades, % from \texttt{Ann},
we need to make inferences in the graph-time domain. Given the above partially observed information,
two such inferred cascades are shown as trees $T_1$ and $T_2$ in
 Fig.~\ref{fig-exa-cascade}. $T_1$ illustrates a cascade where
 each path from the source \texttt{Ann} to each observed node
 has a length that exactly equals to the time step, at which the observed node is influenced,
 %each path
 %from the source \texttt{Ann} to each observed node has the {\em length}
 %exactly equals to the time step it is influenced;
 while $T_2$ illustrates a
 cascade where any path in $T_2$ from \texttt{Ann}
 to an observed node has a length no greater than
 the observed time step when the node is influenced, due to possible delay in observation,
 \eg~\texttt{Mary} is known to be influenced by (instead of exactly at) time step $3$. % instead of exactly at time step 3.
 The inferred cascades provide useful information about the missing
 links and users that are important in the propagation of the information.
\eat{
Motivation example.
Show (1) cascade may be incomplete,
(2) the missing part is important,
(3) cascade inference, and motivate the consistent tree
(4) motivate independent cascade model.
}
\end{example}

The above example highlights the need to make reasonable inference about the cascades, according to
only the partial observations of influenced nodes and the time at or by which they are influenced.
%on nodes and temporal information, \eg the time at or by which nodes are influenced.
%To model the cascades behavior, various cascade models are proposed, \eg~\cite{Granovetter:1978a,Bailey:1975a,kimura2006tractable,NEWMAN:2003a}.
%However,
Although cascade models and a set of related problems, \eg influence maximization,
have been widely studied,
much less is known on how to infer the cascade structures,
including complexity bounds and approximation algorithms.

\stitle{Contributions}. We investigate the cascade inference problem,
where cascades follow the widely used {\em independent cascade model}.
To the best of our knowledge, this is
the first work towards inferring cascades as {\em general trees} following
independent cascade model, based on the partial observations.
% as nodes and the temporal information.

\stab
(a) We introduce the notions of {\em (perfect and bounded) consistent trees} in Section~\ref{sec:preliminary}.
These notions capture the inferred cascades by incorporating connectivity
and time constraints in the partial observations.
To provide a quantitative measure of the quality of
inferred cascades, we also introduce two metrics in Section~\ref{sec:preliminary}, based on
(i) the size of the consistent trees, and
(ii) the likelihood when a diffusion function of
the network graph is taken into account, respectively.
These metrics give rise to two optimization problems,
referred to as the {\em minimum consistent tree} problem
and {\em minimum weighted consistent tree} problem.
 %as the inferred structures of
%a cascade based on partial observation, which only contains
%a set of nodes and the time they are influenced.
%We introduce two types of consistent trees, namely, {\em perfect consistent trees}
%and {\em bounded consistent trees}.
%Based on the consistent trees, we formulate the general cascade inference problem.
%Given a graph and the partial observation of a cascade,
%the problem is to identify consistent trees as a representation of the complete
%cascade which best conform to the observation.
%We introduce two specifications of
%the problem, namely, {\em perfect consistent tree problem}, and {\em bounded consistent tree}
%problem.

\stab
(b) We investigate the problems of identifying
perfect and bounded consistent trees, for given partial observations, in Section~\ref{perfect-infer}
and Section~\ref{bounded-infer}, respectively. These problems are variants of the inference problem.

%complexity bounds of the related problems for perfect consistent trees, and bounded consistent trees as variants
%of the inference problem in Section~\ref{perfect-infer} and Section~\ref{bounded-infer},
%respectively.
%{\em bounded consistent tree problem} as a specification of the
%general cascade predication problem, which is to identify a bounded consistent tree conforming
%to the given partial observation.

%\begin{itemize}
%\item
\stab
(i) We show that these problems are all \NP-complete.
Worse still, the optimization problems are hard to
approximate: unless {\sc p} = \NP, it is not possible to
approximate the problems within any {\em constant} ratio.

\stab
(ii) Nevertheless, we provide approximation and heuristic algorithms
for these problems. % with {\em performance guarantees}.
%\begin{itemize}
%\item
For bounded trees, the problems are $O(|X|*\frac{\log f_{min}}{\log f_{max}})$-approximable, where $|X|$ is the size of the partial observation, and
%$\alpha$ is a function only determined by
$f_{min}$ (resp. $f_{max}$) are the minimum (resp. maximum) probability on the graph edges.
%when a diffusion function is considered. %the ratio of the maximum and minimum propbability on the edges of the social graph.
We provide such polynomial approximation algorithms.
%\item
For perfect trees, we show that it is already
\NP-hard to even find a feasible solution. However, we
provide an efficient heuristics using a greedy strategy.
%\item
Finally, we address a practical special case for perfect tree problems, which
are $O(d*\frac{\log f_{min}}{\log f_{max}})$-approximable, where $d$ is the diameter of the graph, which is typically small in practice.
%\end{itemize}

\stab
(c) We experimentally verify the effectiveness
and the efficiency of our algorithms in Section~\ref{experiment},
using real-life data and synthetic data.
We show that our inference algorithms
can efficiently infer cascades with
satisfactory accuracy.

\stitle{Related work}. We categorize related work as follows.

\etitle{Cascade Models}.
To capture the behavior of cascades, a variety of cascade models have been proposed~\cite{Bailey:1975a,Goldenberg:2001a, Granovetter:1978a,ICALP:Kempe:2005a,Kermack:1927a},
such as %\emph{linear threshold model}~\cite{Granovetter:1978a},
\emph{Suscepctible/Infected~(SI) model}~\cite{Bailey:1975a},
\emph{decreasing cascade model}~\cite{ICALP:Kempe:2005a}, \emph{triggering model}~\cite{KEMPE_KDD},
\emph{Shortest Path Model}~\cite{kimura2006tractable}, %as well as the {\em non-progressive models}
%such as %the {\em Susceptible/Infected/Susceptible~(SIS) model}~\cite{NEWMAN:2003a}
and the {\em Susceptible/Infected/Recover~(SIR) model}~\cite{Kermack:1927a}.
In this paper,
we assume that the cascades follow the \emph{independent cascade model}~\cite{Goldenberg:2001a},
which is one of the most widely studied models (the shortest path model~\cite{kimura2006tractable} is one of its special cases).

%, where the information propagation follows shortest paths).
%These work focus on proposing cascade models %that may,
%rather than inferring cascades based on partial observations. In
%this paper we consider cascade inference problem based on
%independent cascade model.
%to some extend, predicate cascades. In contrast, in this paper we
%investigate the cascade inferencing problem
%based on only partial observations of the cascades.
%assume that all the cascades follows the independent
%cascade model~\cite{Goldenberg:2001a}. %, a widely used model in
%the work of cascade analysis.
%Moreover, we study a special case of
%the cascade inference problem (Section~\ref{perfect-infer})
%over the shortest path model~\cite{kimura2006tractable}.

% need to show the difference
\etitle{Cascade Prediction}.
There has been recent work on cascade prediction and
inference, with the emphasis on global properties (\eg cascade nodes, width, size)~\cite{Budak:2011a:WWW11,2011:Fei:CIKM,kimura2009finding,LappasTGM10,Sadikov:2011a:WSDM11,Song:2007a:WWW,2011:Wang:CoRR}
with the assumption of missing data and partial observations. The problem of identifying and ranking influenced nodes is addressed in~\cite{kimura2009finding,LappasTGM10},
but the topological inference of the cascades is not considered.
%For the spatio-temporal query -- what is the ratio of the active users, $I(d, t)$ (nodes at distance $d$ at time $t$),
Wang et al.~\cite{2011:Wang:CoRR} proposed a \emph{diffusive logistic} model to capture the evolution of the density of active
 users at a given distance over time, and demonstrated the prediction ability of this model.
 Nevertheless, the structural information about the cascade is not addressed. % in~\cite{2011:Wang:CoRR}.
 Song et al.~\cite{Song:2007a:WWW} studied the probability of a user being influenced by a given source.
 %two cascade prediction queries: one is given a user
%$u$ as the source, how probable another user $v$ is influenced; the other
%is to rank users based on how fast users are influenced by a cascade.
In contrast, we consider a more general inference problem where there are multiple observed users,
who are influenced at different time steps from the source.
Fei et al.~\cite{2011:Fei:CIKM} studied social behavior prediction and the effect of information content.
In particular, their goal is to predict actions on an article
%such as replying to an article, no action and other online social actions for online users
based on the training dataset. % consisting of their actions towards articles of different interest or content.
% known cascade info.
Budak et al.~\cite{Budak:2011a:WWW11} investigated the optimization problem of minimizing the number of
the possible influencing nodes following a specified cascade model, instead of predicting
cascades based on partial observations.
\eat{
also considered incomplete observation for influence limitation
problems. In their work, the size of the cascade is fixed and a fixed number of nodes are predicted to
be active so that the likelihood of the observation is maximized, whereas, in our work, we consider a
more general problem. %, where the size of the cascade is unknown.
}

All the above works focus on predicting the nodes and their behavior in the cascades. In contrast,
we propose approaches to infer both the nodes and the topology of the cascades in the graph-time domain.
%either the prediction are , or
%only global properties, or only part of the cascades are predicted or inferred, where the inference
%of the topology and  cascade is not addressed.

\etitle{Network Inference}.
Another host of work study network inference problem, which focuses on inferring network structures from observed cascades
over the unknown network, instead of inferring cascade structures as trees~\cite{dne,RodriguezLK12}.
Manuel et al. ~\cite{RodriguezLK12} proposes techniques to
infer the structure of a network where the cascades flow, based on the observation over
the time each node is affected by a cascade.
Similar network inference problem is addressed in~\cite{dne}, where the cascades are modeled as
(Markov random walk) networks. The main difference between our work and theirs
is (a)
%while the time of the nodes affected by a cascade is assumed to be completely, precisely observed,
we use consistent trees to describe possible cascades allowing partial observations;
(b) we focus on inferring the structure of cascades as trees instead of the backbone networks. % as graphs.
% instead of trees.

Closer to our work is the work by %~\cite{Sadikov:2011a:WSDM11}.
%which is to predict the global cascade properties based on partial observations.
%The prediction of the global cascade properties predication is addressed in~\cite{Sadikov:2011a:WSDM11}
Sadikov et al.~\cite{Sadikov:2011a:WSDM11} that consider %discussed
the prediction of
the cascades modeled as $k$-trees, a balanced tree
model. The global properties of cascades such as size and
depth are predicted based on the incomplete cascade. In contrast to their work, (a) we model
cascades as general trees instead of $k$-balanced trees, (b) while Sadikov et al.~\cite{Sadikov:2011a:WSDM11}
assume the partial cascade is also a $k$-tree and predict only the properties of the
original cascade, we infer the nodes as well as topology
of the cascades only from a set of nodes and their activation time, using much less
available information. (c) The temporal information (\eg time steps) in the partial observations
is not considered in~\cite{Sadikov:2011a:WSDM11}.

\eat{
\etitle{Steiner tree problems}.
~\cite{voss1999steiner,costa2008fast,kantor2006approximate}

Other related.
~\cite{gomez2010inferring};
~\cite{LappasTGM10};
~\cite{Galuba10};
~\cite{dne};
}
%~\cite{Chen10}

\newcommand{\mst}{\kw{MST}}
\newcommand{\mpct}{\kw{PCT_{min}}}
\newcommand{\cmpmpct}{\kw{PCT}}
\newcommand{\mwpct}{\kw{PCT_w}}
\newcommand{\cmpwpct}{\kw{WPCT}}
\newcommand{\cmpwpctsp}{\kw{WPCT_{sp}}}
\renewcommand{\sp}{\kw{sp}}
\newcommand{\lsearch}{\kw{PCT_l}}
\newcommand{\mwdst}{\kw{DST}}
\newcommand{\conbct}{\kw{CT}}
\newcommand{\mbct}{\kw{BCT_{min}}}
\newcommand{\cmpmbct}{\kw{BCT}}
\newcommand{\mwbct}{\kw{BCT_w}}
\newcommand{\cmpwbct}{\kw{WBCT}}
\newcommand{\bfs}{\kw{BFS} \xspace}

\section{Consistent Trees}
\label{sec:preliminary}

We start by introducing several notions. %in this section.

%\subsection{Socail graphs, cascades and partial observations}

%We introduce social graphs and cascades as follows.
% cascades and consistent trees.  We then introduce the
%cascades and independent cascade model. Based on these notions we introduce the {\em consistent cascade}.

\stitle{Diffusion graph}. We denote a social network
as a {\em directed} graph $G = (V, E, f)$, where (a)
$V$ is a finite set of nodes, and each node $u \in V$ denotes a user;
(b) $E \subseteq V \times V$ is a finite set of edges, where each edge
$(u, v) \in E$ denotes a social connection via which the information may diffuse from $u$ to $v$;
and (c) a {\em diffusion function} $f: E \rightarrow R^+$ which assigns for each edge $(u,v) \in E$ a value
$f(u,v) \in [0,1]$, as the probability that node $u$ {\em influences} $v$.% \ie the information diffuses from $u$ to $v$.

\eat{
We shall use the following notations.
(1) A {\em path} $\rho$ in graph $G$ is a sequence of nodes $v_1/
\ldots/ v_n$ such that $(v_i, v_{i+1})$ is an edge in $G$, and $v_i$ = $v_j$ if and only if
$i$ = $j$, for each $i,j\in[1, n-1]$.
(2) The {\em length} of the path $\rho$, denoted by $\len(\rho)$,
is $n-1$, \ie it is the number of edges in $\rho$.
we refer to $v_2$ as a {\em child} of $v_1$
(or $v_1$ as a {\em parent} of $v_2$),
and $v_i$ as a {\em descendant} of $v_1$ for $i \in [2, n]$.
(3) For $u, v \in V$, we say $u$ reaches $v$ by $h$ hops if there exists a path $\rho$ from $u$ to $v$,
such that $\len(\rho) \leq h$.
(4) The {\em distance} from $u$ to $v$,
denoted as $\dist(u,v)$, is the length of the shortest path from $u$ to $v$. % in $G$.
}

 %%%%%%%%%%%%%%%
\begin{figure}[tb!]
%\vspace{-3ex}
\centerline{\includegraphics[scale=0.5]{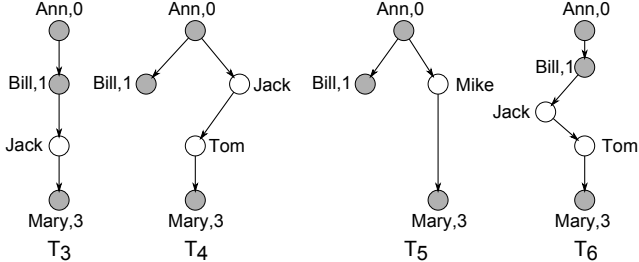}}
%\vspace{-2ex}
%\caption{Consistent trees \wrt partial observation}
\caption{Tree representations of a partial observation $X$ = $\{(\texttt{Ann}, 0)$, $(\texttt{Bill}, 1)$, $(\texttt{Mary}, 3)\}$: $T_3$, $T_4$ and $T_5$ are consistent Trees, while $T_6$ is not.}
\label{exa-contree}
\vspace{-4ex}
\end{figure}
%%%%%%%%%%%%%%%%

% Given a graph $G = (V,E,f)$, a cascade $C$ is a tree $(V_c,E_c,s)$ where $V_c \subseteq V$, $E_c \subseteq E$, and $s \in V_c$
%is the {\em source node} of $C$.

\stitle{Cascades}.
We first review the {\em independent cascade model}~\cite{KEMPE_KDD}.
 %one of the most commonly used cascade model in a wide range of study~\cite{}.
%For example, in the case of Twitter network, $u$ and $v$ are two users, and there is
%an edge $\langle u, v \rangle$ if $v$ follows $u$.
%\stitle{Independent cascade model}~\cite{KEMPE_KDD}.
% where cascades unfold,
%the independent cascade model $M$ = $(G, s,f)$ consists of graph $G$, a {\em source node} $s \in V$, a sequence of timestamps $T$ = $\{t_0, \ldots, t_i\}$,
%and a difussion function $f$ which assigns for each edge $(u,v) \in E$ a probability that node $u$ {\em activates} node $v$, \ie the information diffuses from
%$u$ to $v$.
%Given a graph $G$,
We say an information propagates over a graph $G$ following the {\em independent cascade model} if
(a) at any time step, each node in $G$ is exactly one of the three states $\{$\emph{active}, \emph{newly active}, \emph{inactive}$\}$;
%\item there are three states for nodes: \emph{active}, \emph{newly active} and \emph{inactive} and any node has only one state at any time step;
(b) a cascade starts from a {\em source node} $s$ being \emph{newly active} at time step $0$;
(c) a \emph{newly active} node $u$ at time step $t$ has only one chance to influence its \emph{inactive} neighbors,
such that at time $t+1$, (i) if $v$ is an inactive neighbor of $u$, $v$ becomes \emph{newly active} with probability $f(u, v)$; and
(ii) the state of $u$ changes from \emph{newly active} to \emph{active}, and cannot influence any neighbors afterwards; and
(d) each {\em inactive} node $v$ can be influenced by at most one of its \emph{newly active} neighbors independently, and the neighbors' attempts are sequenced in an arbitrary order. % in an arbitrary order/
%In the above process,
Once a node is \emph{active}, it cannot change its state.  %turn back to being \emph{inactive} or \emph{newly active}.
%As remarked earlier, the process of the information propagation over a graph often follows a specified cascade model.
%Before we define cascades,
%(e) the process ends if there are no newly active nodes.
\eat{
\begin{itemize}
\item at any time step, each node in $G$ has exactly one of the three states in $\{$ \emph{active}, \emph{newly active}, \emph{inactive} $\}$;
%\item there are three states for nodes: \emph{active}, \emph{newly active} and \emph{inactive} and any node has only one state at any time step;
\item a cascade starts from a {\em source node} $s$ being \emph{newly active} at time step $0$.
\item for a \emph{newly active} node $u$ at time step $T$, $u$ has only one opportunity to activate its \emph{inactive} neighbors,
such that at time $T+1$, (a) if $v$ is an inactive neighbor of $u$, $v$ becomes \emph{newly active} with probability $f{u, v}$;
(b) the state of $u$ changes from \emph{newly active} to \emph{active}, and cannot activate any \emph{inactive} neighbors afterwards;
once a node is \emph{active}, it cannot turn back to being \emph{inactive} or \emph{newly active};
\item each {\em inactive} node $v$ is activated by at most one of its multiple \emph{newly active} neighbors independently; % in an arbitrary order/
\item the process ends if there are no newly active nodes.
\end{itemize}
}

Based on the independent cascade model, we define a {\em cascade} $C$ over graph $G$ = $(V,E,f)$
as a {\em directed tree} $(V_c,E_c,s, {\cal T})$ where
(a) $V_c \subseteq V$, $E_c \subseteq E$;
(b) $s \in V_c$ is the {\em source node} from which the information starts to propagate; and
(c) ${\cal T}$ is a function which assigns for each node $v_i \in V_c$ a {\em time step} $t_i$, which represents
that $v_i$ is {\em newly active} at time step $t_i$.
Intuitively, a cascade is a tree representation of the
``trace'' of the information propagation from a specified source node $s$ to a set of influenced nodes.

Indeed, one may verify that any cascade from $s$ following the independent cascade model is a tree rooted at $s$.
%In this paper we assume that all the cascades follow the independent cascade model. %One may observe that

\begin{example}
\label{exa-icm}
The graph $G$ in Fig.~\ref{exa-cascade} depicts a social graph.
The tree $T_1$ and $T_2$ are two possible cascades following the independent cascade model.
For instance, after issuing an ad of ``Iphone 4s'', \texttt{Ann} at time $0$
becomes ``newly active''.
\texttt{Bill} and \texttt{Jack} retweet the ad at time $1$. \texttt{Ann} becomes ``active'',
while \texttt{Bill} and \texttt{Jack} are turned to ``newly active''. The process
repeats until the ad reaches \texttt{Mary} at time step $3$. The trace of the information propagation
forms the cascade $T_1$.
\end{example}

%is the {\em source node} of $C$.

As remarked earlier, it is often difficult
to observe the entire structure of a cascade in practice. We model the
observed information for a cascade as a {\em partial observation}.

\stitle{Partial observation}. Given a cascade $C$ = $(V_c,E_c,s,{\cal T})$,
a pair $(v_i, t_i)$ is an \emph{observation point}, if $v_i \in V$ is known (observed) to be \emph{newly active} {\em at} or {\em by} time step $t_i$.
A {\em partial observation} $X$ is a set of observation points. %, as well as an observation point $(s, 0)$.
Specifically, $X$ is a {\em complete observation} if for any $v \in V_c$, there is %exactly one
an observation point $(v, t) \in X$. %Otherwise, $X$ is a {\em partial observation}.
To simplify the discussion, we also assume that pair $(s, 0) \in X$ where $s$ is the source node.
The techniques developed in
this paper can be easily adapted to the case where the source node is unknown.

\eat{
Intuitively, an observation point $(v, t)$ indicates that
node $v$ is newly activate at time stamp $t' \leq t$, due to possible delay of the observation.
To simplify the discussion, we assume that pair $(s, 0) \in X$ where $s$ is the source node. As will be discussed, the techniques developed in
this paper can easily be adapted to the case where the source node is unknown.
}

\eat{
\etitle{Remarks.} Observe the following. (1) To simplify the discussion, we assume that pair $(s, 0) \in X$. As will be discussed, the techniques developed in
this paper can easily be adapted to the case where the source node is unknown.
(2) The observation pair $(v_i, t_i)$ indicates that
node $v_i$ is newly activate at some time stamp $t \leq t_i$, due to possible delay in the observed information.
}

We are now ready to introduce the idea of consistent trees. %Specifically, we introduce {\em perfect consistent trees} and {\em bounded consistent trees}.

\subsection{Consistent trees}

Given a partial observation $X$ of a graph $G$ = $(V,E,f)$, a {\em bounded consistent tree} $T_s$ = $(V_{T_s},$ $E_{T_s}, s)$ \wrt $X$ is
a directed subtree of $G$ with root $s \in V$, such that for {\em every} $(v_i, t_i) \in X$, $v_i \in V_{T_s}$, and $s$ reaches $v_i$ by $t_i$ hops, \ie
there exists a path of length {\em at most} $t_i$ from $s$ to $v_i$.
Specifically, we say a consistent tree is a {\em perfect consistent tree} if for every $(v_i, t_i) \in X$ and  $v_i \in V_{T_s}$,
there is a path of length {\em equals to} $t_i$ from $s$ to $v_i$. %Otherwise, it is a {\em bounded consistent tree}.

Intuitively, consistent trees represent possible cascades which
conform to the independent cascade model, as well as the partial observation.
%Observe the following.
Note the following:
 %Indeed,
%(1) For each observation point $(x_i,t_i)$, each path from $s$ to $x_i \in X$ in a bounded (resp. perfect) consistent tree $T$
%denotes a trace of the information from time step $0$ to $t_j \leq$ (resp. $=$) $t_i$. Moreover,
(a) the path from the root $s$ to a node $v_i$ in a bounded consistent tree $T_s$ is not necessarily
a shortest path from $s$ to $v_i$ in $G$, as observed in~\cite{kossinets2008structure};
%
%(2) The consistent trees as inferred cascades follow the independent cascade model. %Each node in
%the consistent trees cannot be activated twice (\ie has only one parent).
%
(b) % As remarked earlier, the observed time step $t_i$ in $(v_i,t_i) \in X$ may actually be larger than the time $t$
%when $v_i$ is newly activated.
%Treating $X$ as the only available information for cascade inference, (a)
the perfect consistent trees model
cascades when the partial observation is accurate, \ie each time $t_i$ in an observation point
$(v_i,t_i)$ is exactly the time when $v_i$ is newly active; in contrast, in bounded consistent trees,
%model the cascades where
 %the partial observation may not be accurate, \ie
an observation point $(v, t)$ indicates that
node $v$ is newly active at the time step $t' \leq t$, due to possible {\em delays} in the information propagation,
as observed in~\cite{ChaBAG12}.  % as verified in Section~\ref{experiment}.
%On the other hand,

%In other words,
%given a graph $G$ and an observation $X$, the consistent trees represent all the possible cascades conforming to $X$ in $G$
%which follow the independent cascade model.

\begin{example}
\label{exa-consist}
Recall the graph $G$ in Fig.~\ref{exa-cascade}.
The partial observation of a cascade in $G$ is $X$ = $\{(\texttt{Ann},0)$, $(\texttt{Bill},1)$, $(\texttt{Mary},3)\}$.
The tree $T_1$ is a perfect consistent tree \wrt $X$, where $T_2$ is a bounded consistent tree \wrt $X$.

Now consider the trees in Fig.~\ref{exa-contree}. One may verify that
(a) $T_3$, $T_4$ and $T_5$ are bounded consistent trees \wrt~$X$;
(b) $T_3$ and $T_4$ are perfect consistent trees \wrt~$X$, where $T_5$ is not a perfect consistent tree.
(c) $T_6$ is not a consistent tree, as there is no path from the source \texttt{Ann} to \texttt{Mary} with length
no greater than $3$ as constrained by the observation point $(\texttt{Mary},3)$.
\end{example}

\eat{
\begin{example}
In Figure~\ref{fig:example}, $G = (V, E)$ is a directed graph of six nodes and ten edges. The diffusion function $f$ map each edge to a real number shown as the probability around each edge. $X$ is an observation of a cascade with two \emph{observation points}: $\{(b, 1), (d, 3)\}$. In this example, there are $8$ possible consistent trees rooted from $a$ indexed from $T^0_a$ to $T^7_a$.
\end{example}

\begin{figure*}[ht]
\centering
\subfigure[$G = (V, E)$]{
\includegraphics[scale=0.1]{figure/example/case1-in}
\label{fig:example:case1in}
}\subfigure[$T^0_a$]{
\includegraphics[scale=0.1]{figure/example/case1tree0}
\label{fig:example:case1tree0}
}\subfigure[$T^1_a$]{
\includegraphics[scale=0.1]{figure/example/case1tree1}
\label{fig:example:case1tree1}
}\subfigure[$T^2_a$]{
\includegraphics[scale=0.1]{figure/example/case1tree2}
\label{fig:example:case1tree2}
}\subfigure[$T^3_a$]{
\includegraphics[scale=0.1]{figure/example/case1tree3}
\label{fig:example:case1tree3}
}
\subfigure[$T^4_a$]{
\includegraphics[scale=0.1]{figure/example/case1tree4}
\label{fig:example:case1tree4}
}\subfigure[$T^5_a$]{
\includegraphics[scale=0.1]{figure/example/case1tree5}
\label{fig:example:case1tree5}
}\subfigure[$T^6_a$]{
\includegraphics[scale=0.1]{figure/example/case1tree6}
\label{fig:example:case1tree6}
}\subfigure[$T^7_a$]{
\includegraphics[scale=0.1]{figure/example/case1tree7}
\label{fig:example:case1tree7}
}
\caption{An example of Graphs, Cascades and Consistent Trees}
\label{fig:example}
\end{figure*}
}

\subsection{Cascade inference problem}

We introduce the general cascade inference problem.
Given a social graph $G$ and a partial observation $X$,
the {\em cascade inference problem} is to determine whether there exists
a consistent tree $T$ \wrt $X$ in $G$.

There may be multiple consistent trees for a partial observation,
%As there may exists multiple consistent trees as valid inferred
%cascades,
so one often wants to identify the best consistent tree.
We next provide two quantitative metrics to measure the quality of
the inferred cascades. Let $G$ = $(V,E,f)$ be a social graph,
and $X$ be a partial observation.
%In particular,
%we denote the {\em minimum bounded (resp. perfect) consistent tree} problem
%as~\

\stitle{Minimum weighted consistent trees}. In practice, one often wants to % Alternatively, o
identify the consistent trees that are most likely to be the
real cascades. %We consider the {\em maximum likelihood} of the consistent trees.
Recall that each edge $(u, v)$ $\in E$ in a given network $G$ carries a value assigned by
a diffusion function $f(u, v)$, which indicates the probability that $u$ influences $v$. %upon \eg receiving a message.
Based on $f(u, v)$, we introduce a {\em likelihood function} as a quantitative metric for consistent trees.

\etitle{Likelihood function}. Given a graph $G$ = $(V,E,f)$, a partial observation
$X$ and a consistent tree $T_s = (V_{T_s}, E_{T_s},s)$, the
{\em likelihood} of $T_s$, denoted as $L_X(T_s)$,
is defined as:

\vspace{-1.5ex}
\begin{equation}
L_X(T_s) = \mathbb{P}(X \mid T_s) = \prod_{(u, v) \in E_{T_s}} f(u, v).
\label{eq:mlct}
\end{equation}
\vspace{-.5ex}

Following common practice, we opt to use the log-likelihood metric, where % may need bib here.
\vspace{-.5ex}
\[
L_X(T_s) = \sum_{(u, v) \in E_{T_s}} \log f(u, v)
\]
\vspace{-.5ex}

%Therefore,
Given $G$ and $X$, a natural problem is to find the consistent tree of
the maximum likelihood in $G$ \wrt $X$. Using log-likelihood, the {\em minimum weighted consistent tree} problem
is to identify the consistent tree $T_s$ with the minimum $- L_X(T_s)$, which in turn has the
maximum likelihood.

\stitle{Minimum consistent trees}.
Instead of weighted consistent trees,
%This metric, as a special case of the
%minimum weighted consistent tree, evaluates the
%{\em size} of a consistent tree as the size of its
%edge set. Indeed, in practice
one may simply want to find the {\em minimum}
structure that represents a cascade~\cite{mathioudakis2011sparsification}.
The minimum consistent tree, as a special case of the
minimum weighted consistent tree, depicts the smallest
cascades with the fewest communication steps
to pass the information to all the observed nodes.
In other words, the metric favors those consistent trees consist with the
given partial observation with the fewest edges.

%{\em size} of a consistent tree as its
%edge number.
%as inferences for a partially observed cascade.
%Intuitively,
%this metric depicts the smallest

Given $G$ and $X$, the {\em minimum consistent tree} problem is to find the
minimum consistent trees in $G$ \wrt $X$. %as the inferred cascade structures.

In the following sections, we
investigate the cascade inference problem, and the related
optimization problems using the two metrics.
We investigate the problems for perfect consistent trees in Section~\ref{perfect-infer},
and for bounded consistent trees in Section~\ref{bounded-infer},
respectively.

%In the rest part of the paper, we introduce the techniques to inference
%the cascade based on the consistent trees, using only the partial observation.
% relying only on the size of the observable part.

\eat{
In this work, we are interested in discovering unobserved nodes that result in the partial observation of a cascade.
Given a social network and a partial observation, there are multiple ways to develop the objective function such as
finding $k$ nodes that maximize the likelihood of the observation, finding $k$ sources that maximize the likelihood
of the observation, finding the consistent tree with maximum likelihood and so on. In this paper, we focus on finding
the consistent tree that maximizes the likelihood.
}

\section{Cascades as perfect trees}
\label{perfect-infer}

%In this section we introduce a specification of the general cascade inference problem, where we
%constraint the cascade as perfect consistent trees. We first formulate the (minimum) perfect consistent
%tree problem. % in Section~\ref{sec-bct}.
 %In Section~\ref{sec-mnwbct}
%we introduce a metric to measure the quality of the perfect consistent trees, and
%the  as well
%as their corresponding optimization problems, respectively.

%In this section we investigate the cascade inference based on perfect consistent trees.
As remarked earlier, when the partial observation $X$ is accurate,
\eat{
where
each observation point $(v,t) \in X$ exactly describe the newly activate time step $t$ for a node $v$,
}
one may want to infer the cascade structure via {\em perfect consistent trees}.
%In this section we first investigate cascade inference problems based on perfect consistent trees.
\eat{
and establish complexity bounds and algorithms for the problems.
We then consider a practical special case of the problem, and provide approximation algorithms
with provable performance guarantees.
}
%\subsection{Inferring perfect consistent trees}
%Given a social graph $G$ and a partial observation $X$, the {\em perfect consistent tree} problem,
%denoted as \pct, is to find a perfect consistent tree in $G$ \wrt $X$.
%Along the same line in Section~\ref{sec:preliminary},
The {\em minimum (resp. weighted) perfect consistent tree} problem,
denoted as \mpct (resp. \mwpct) is to find the perfect consistent trees with minimum size (resp. weight)
as the quality metric. % illustrated in Section~\ref{bounded-infer}.

%Given a social graph $G$ and a partial observation $X$,
%One may want to identify a polynomial time algorithm to identify perfect consistent trees.
Though it is desirable to have efficient polynomial time algorithms to identify perfect consistent trees, the problems of searching \mpct and \mwpct are nontrivial.
%It is desirable if the \mpct and \mwpct problems can be solved in polynomial time. 
%Nevertheless, these problems are nontrivial.
%Nevertheless, it is more difficult to find a perfect consistent tree, in contrast to its bounded counterpart (Proposition~\ref{bct}).
%We first investigate the problem on determine whether there exists a perfect consistent tree for
%a given social graph $G$ and a partial evaluation $X$.

\vspace{1ex}
\begin{proposition}
\label{pct}
Given a graph $G$ and a partial observation $X$,
%\begin{itemize}
%\vspace{-0.5ex}
%\item[(1)]
(a) it is \NP-complete to determine whether there is a perfect consistent tree \wrt $X$
in $G$; and
%\vspace{-0.5ex}
%\item[(2)]
(b) the~\mpct and ~\mwpct problems are \NP-complete and \APX-hard.
%\end{itemize}
\end{proposition}
\vspace{1ex}

One may verify Proposition~\ref{pct}(a) by a reduction from the Hamiltonian path problem~\cite{Vazirani:2001:book}, which is to determine
whether there is a simple path of length $|V|-1$ in a graph $G$ =$(V,E)$.
%Proposition~\ref{pct} shows that it is already intractable to determine whether
%there exists a perfect consistent tree \wrt a partial observation.
Following this, one can verify that the \mpct and \mwpct problems are \NP-complete as an immediate result.

Proposition~\ref{pct}(b) shows that the \mpct and \mwpct problems are hard to approximate.
The \APX class~\cite{Vazirani:2001:book} consists of \NP optimization problems that can be
approximated by a polynomial time (\PTIME) algorithm within {\em some} positive constant.
The \APX-hard problems are \APX problems to which every \APX problem
can be reduced.
Hence, the problem for computing a minimum (weighted) perfect consistent tree
is among the %is intractable and moreover, it
hardest ones %in the class of problems
that allow \PTIME algorithms with a constant approximation ratio.

It is known that if there is an {\em approximation preserving reduction (\AFPR)}~\cite{Vazirani:2001:book} from a problem
$\Pi_1$ to a problem $\Pi_2$, and if problem $\Pi_1$ is
\APX-hard, then $\Pi_2$ is \APX-hard~\cite{Vazirani:2001:book}.
To see Proposition~\ref{pct}(b), we may construct an \AFPR from the minimum directed steiner tree (\mst) problem.
An instance of a directed steiner tree problem $I$ = $\{G, V_r, V_s, r, w\}$ consists of a graph $G$,
a set of {\em required} nodes $V_r$, a set of {\em steiner} nodes $V_s$, a source node $r$ and a function
$w$ which assigns to each node a positive weight. The problem is to find a minimum weighted tree rooted at $r$, such that
it contains all the nodes in $V_r$ and a part of $V_s$. We show such a reduction exists.
%It is known that
Since \mst is \apx-hard, %~\cite{Vazirani:2001:book},
%Thus,
\mpct is \apx-hard.
%known to be \APX-hard~\cite{Vazirani:2001:book}.

%\begin{proposition}
%\label{mwpct}

%\end{proposition}

\subsection{Bottom-up searching algorithm}

%Despite of
Given the above intractability and approximation hardness result,
%in this section
we introduce a heuristic %algorithm, denoted as
~\cmpwpct for the \mwpct problem. The idea is to
(a) generate a ``backbone network'' $G_b$ of $G$ which contains all the nodes and edges that are possible to
form a perfect consistent tree, using a set of {\em pruning rules}, and also rank the observed nodes in $G_b$
with the descending order of their time step in $X$, and (b) perform a bottom-up evaluation for each time step in $G_b$
using a local-optimal strategy, following the descending order of the time step.

\stitle{Backbone network}. We consider pruning strategies to reduce the nodes and the edges that are not possible to
be in any perfect consistent trees, given a graph $G$ = $(V,E,f)$ and a partial observation $X$ = $\{(v_1,t_1), \ldots, (v_k,t_k)\}$.
We define a backbone network $G_b$ = $(V_b, E_b)$, where % which is a subgraph of $G$ where
\begin{itemize}
\item $V_b$ = $\bigcup \{v_j|\dist(s,v_j) + \dist(v_j, v_i)\leq  t_i \}$ for each $(v_i, t_i) \in X$; and
\item $E_b$ = $\{(v',v)| v' \in V_b, v \in V_b, (v',v) \in E \}$
\end{itemize}

Intuitively, $G_b$ includes all the possible nodes and edges that may appear in a perfect consistent tree for a given
partial observation. In order to construct $G_b$, a set of {\em pruning rules} can be developed as follows:
 if for a node $v'$ and each observed node $v$ in a cascade with time step $t$,  $\dist(s,v') + \dist(v',v) > t$,
 then $v'$ and all the edges connected to $v'$ can be removed from $G_b$.

%%%%%%%%%%%%%%%%%%%%%%%%%%%%%%%%%%%%%%%%
\begin{figure}[tb!]
\vspace{-1.5ex}
\begin{center}
\fbox{
{\small
%\fcolorbox{black}{yellow}{
%\fbox{
\begin{minipage}{3.3in}
%\myhrule
%\vspace{-2.5ex}
\mat{0ex}{
%{\bf Algorithm}~\pSim\\
\sstab {\sl Input:\/} \= graph $G$ and partial observation $X$.\\
{\sl Output:\/} a perfect consistent tree $T$ in $G$. \\
%\>  /*Initialization*/ \\
\stab \bcc \hspace{1ex} \= tree $T$ = $(V_T,E_T)$, where $V_T$ := $\{v| (v,t) \in X\}$,\\
\>  set level $l(v)$:= $t$ for each $(v,t)\in X$, $E$ := $\emptyset$; \\
\icc \> set $V_b$ := $\{v_b| \dist(s,v_b) \leq t_{max} \}$; \\ %where $t_{max}$ is the largest time step in $X$;\\
\icc \> \If there is a node $v$ in $X$ and $v \notin V_b$ \Then \Return $\emptyset$;\\
\icc \> set $E_b$ := $\{(v',v)|(v',v)\in E, v'\in V_b, v\in V_b)\}$; \\%as the node induced subgraph by $V_b$ in $G$;\\
\icc \> \For \Each $v \in V_b$ \Do \\
\icc \> \hspace{1ex} \If there is no $(v_i,t_i) \in X$ that \\
\>\hspace{1ex} $\dist(s,v)$+$\dist(v,v_i) \leq t_i$ \Then \\
\icc \> \hspace{1ex} $V_b$ = $V_b \setminus \{v\}$; \\
\icc \> \hspace{1ex} $E_b$ = $E_b \setminus \{(v_1,v_2)\}$ where $v_1 = v$ or $v_2 = v$; \\
\icc \> graph $G_b$ := $(V_b, E_b)$; \\
 %\> /*sorting*/\\
\icc\> list $L$ := $\{(v_1,t_1),\ldots, (v_k,t_k)\}$ \\
\> where $t_i \leq t_{i+1}$, $(v_i,t_i) \in X$, $i \in [1,k-1]$; \\
%\> /*bottom-up construction*/ \\
\icc \> \For \Each $i$ $\in [1, t_{max}] $ following descending order \Do \\
%\icc \> \hspace{1ex} construct graph $G_t$ as node induced subgraph of $G_b$ by $V_t$, where \\
\icc\> \hspace{1ex}  $V_t$:= $V_1 \cup V_2 \cup V_3$, $V_1$ := $\{v_i| (v,t_i)\in X\}$; \\
\>\hspace{1ex}  $V_2$ := $\{v| v \in V_T, l(v) = t_i \}$;\\ %\\where $l(v)$ is the topological order of $v$ \in $T$; \\
\>\hspace{1ex}  $V_3$ := $\{v'|(v',v)\in E_b, v\in V_1 \cup V_2, v'\notin V_T\}$; \\
\icc \> \hspace{1ex} $E_t$ := $\{(v',v)| v' \in V_3, v \in V_1\cup V_2, (v',v) \in E_b \}$; \\
\icc \> \hspace{1ex} construct $G_t$ = $(V_t, E_t)$; \\
%\>\>  $V_4$ := $\{v''| (v)\}$
\icc \>\hspace{2ex} $T$ := $T \cup \lsearch(G_t, V_1 \cup V_2, V_3, i)$; \\
%\icc \> \hspace{2ex} $G_b$ := $G_b \setminus T$; \\
\icc \> \If $T$ is a tree \Then \Return $T$; \\
\icc \> \Return $\emptyset$; \\
}

\vspace{-4ex}
\mat{0ex}{
{\bf Procedure}~\lsearch\\
\sstab {\sl Input:\/} \= A bipartite graph $G_t$, \\
\> node set $V$, node set $V_s$, a number $t_i$; \\
{\sl Output:\/} a forest $T_t$.\\
\sstab \bcc \hspace{1ex}\= $T_t$ = $\emptyset$; \\
%\icc \>
\icc \> construct $T_t$ as a minimum weighted steiner forest \\
\> which cover $V$ as the required nodes; \\
\icc \> \For \Each tree $T_i \in T_t$ \Do \\
\icc \> \hspace{2ex} $l(r)$ := $t_i - 1$ where $r \in V_s$ is the root of $T_i$; \\
\icc \> \Return $T_t$;\\
}
\vspace{-4.5ex}
\myhrule
\end{minipage}
}
}
\end{center}
\vspace{-2ex}
%\vspace{-2.5ex}
\caption{Algorithm~\cmpwpct: initialization, pruning and local searching}
\label{fig-cmpwpct}
\vspace{-4ex}
%\vspace{-2ex}
\end{figure}
%%%%%%%%%%%%%%%%%%%%%%%%%%%%%%%%%%%%%

\stitle{Algorithm}. Algorithm~\cmpwpct, as shown in Fig.~\ref{fig-cmpwpct},
consists of the following steps:

%\stab
\etitle{Initialization} (line~1). The algorithm~\cmpwpct starts by initializing a tree $T$,
by inserting all the observation points into $T$. Each node $v$ in $T$ is assigned with a {\em level}
$l(v)$ equal to its time step as in $X$. The edge set is set to empty.

%\stab
\etitle{Pruning} (lines~2-10). The algorithm~\cmpwpct then constructs a backbone network $G_b$ with the pruning rules (lines~2-9).
It initializes a node set $V_b$ within $t_{max}$ hop of the source node $s$, where $t_{max}$
is the maximum time step in $X$ (line~2). If there exists some node $v \in X$ that is not in $V_b$, the algorithm returns $\emptyset$, since there is no path from $s$ reaching $v$ with $t$ steps for $(v,t) \in X$ (line~3). It further removes the redundant nodes and edges that are not in any perfect trees, using the pruning rules (lines~5-8). The network $G_b$ is then constructed with $V_b$ and $E_b$ at line~9.
The partial observation $X$ is also sorted \wrt the time step (line~10). % of each observation point .

%\stab
\etitle{Bottom-up local searching} (lines~11-17).
Following a bottom-up greedy strategy, the algorithm~\cmpwpct processes each observation point as follows.
For each $i$ in $[1,t_{max}]$, it generates a (bipartite) graph $G_t$.
(a) It initializes a node set $V_t$ as the union of three sets of nodes $V_1$,
$V_2$ and $V_3$ (line~12), where (i) $V_1$ is the nodes in the observation points with time step $t_i$,
(ii) $V_2$ is the nodes $v$ in the current perfect consistent tree $T$ with level $l(v)$ = $t_i$,
and (iii) $V_3$ is the union of the parents for the nodes in $V_1$ and $V_2$.
(b) It constructs an edge set $E_t$ which consists of the edges from the nodes in
$V_3$ to the nodes in $V_1$ and $V_2$. (c) It then generates $G_t$ with $V_t$ and the edge set $E_t$,
which is a bipartite graph.
After $G_t$ is constructed, the algorithm~\cmpwpct invokes procedure~\lsearch to compute a ``part''
of the perfect tree $T$, which is an {\em optimal} solution for $G_t$, a part of the graph $G_b$
which contains all the observed nodes with time step $t_i$.
It expands $T$ with the returned partial tree (line~15). % and removes
%from $G_b$ those nodes and edges that are already added into the tree $T$ (line~16).
The above process (lines~11-15) repeats for each $i \in [1, t_{max}]$ until
all the nodes in $X$ are processed. Algorithm~\cmpwpct then checks if the
constructed $T$ is a tree. If so, it returns $T$ (line~16).
Otherwise, it returns $\emptyset$ (line~17).
The above procedure is as illustrated in Fig.~\ref{exa-alg}.

\etitle{Procedure~\lsearch}. Given a (bipartite) graph $G_t$, and two sets of nodes $V$ and $V_s$ in $G_t$,
the procedure~\lsearch computes for $G_t$ a set of trees $T_t$ = $\{T_1, \ldots, T_i\}$ with the minimum total weight (line~2),
such that (a) each $T_i$ is a $2$-level tree with a root in $V_s$ and leaves in $V$,
(b) the leaves of any two trees in $T_t$ are disjoint, and
(c) the trees contain all the nodes in $V$ as leaves.
For each $T_i$, \lsearch assigns its root $r$ in $V_s$ a level $l(r)$ = $t_i-1$ (line~4).
$T_t$ is then returned as a part of the entire perfect consistent tree (line~5).
In practice, we may either employ linear programming, or an
 algorithm for \mst problem (\eg~\cite{RobinsZ05}) to compute $T_t$.

\vspace{-1ex}
 %%%%%%%%%%%%%%%
\begin{figure}[tb!]
%\vspace{-3ex}
\centerline{\includegraphics[scale=0.35]{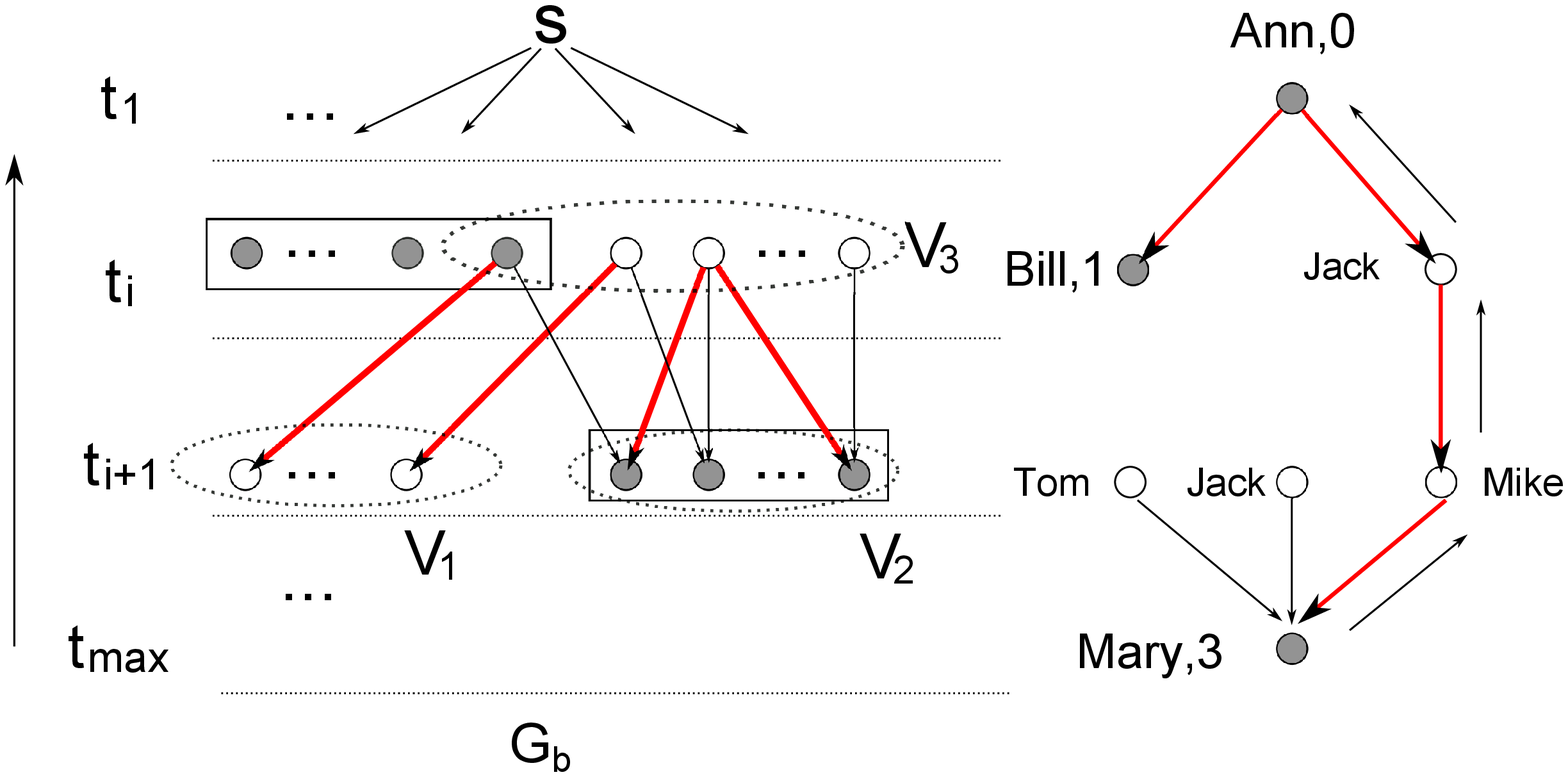}}
%\vspace{-2.5ex}
%\caption{Local search strategy over backbone network}% $G_b$
\caption{The bottom-up searching in the backbone network}
%A $2$-level forest, as a part of a perfect tree, is constructed between level $t_i$ and $t_{i+1}$.  }
\label{exa-alg}
\vspace{-3ex}
\end{figure}
%%%%%%%%%%%%%%%%

\vspace{1ex}

\begin{example}
\eat{
For the observed nodes $V_2$ at time step $t_{i+1}$, the algorithm~\cmpwpct
identifies $V_1$ as the nodes in $T$ with the same level, and $V_3$ as the union of the
parents of $V_1$ and $V_2$. For the induced graph $G_t$ which consists of $V_1$, $V_2$
and $V_3$ as well as the edges from $V_3$ to the nodes in $V_1$ and $V_2$, the optimal
steiner forest is selected as shown in the figure, as a part of a perfect consistent tree.
}
%
%Recall the cascade $T_1$ in Fig.~\ref{exa-cascade}.
The cascade $T_1$ in Fig.~\ref{exa-cascade}, as a minimum weighted perfect consistent tree,
can be inferred by algorithm~\cmpwpct as illustrated in Fig.~\ref{exa-alg}.
~\cmpwpct first initializes a tree $T$ with the node $\texttt{Mary}$.
It then constructs $G_t$
as the graph induced by edges $(\texttt{Tom, Mary})$, $(\texttt{Jack, Mary})$, and $(\texttt{Mike, Mary})$.
Intuitively, the three nodes as the parents of $\texttt{Mary}$ are the possible nodes which
accepts the message at time step $2$.
It then selects the tree with the maximum probability, which is a single edge $(\texttt{Mike, Mary})$,
and adds it to $T$. Following $\texttt{Mike}$, it keeps choosing the optimal tree structure for each level, and identifies
nodes $\texttt{Jack}$. The process repeats until
~\cmpwpct reaches the source $\texttt{Ann}$. It then returns the perfect consistent tree $T$ as the
inferred cascade from the partial observation $X$.
\end{example}

%\vspace{-1ex}
\etitle{Correctness}. The algorithm~\cmpwpct either returns $\emptyset$,
or correctly computes a perfect consistent tree \wrt the partial observation $X$.
Indeed, one may verify that (a) the pruning rules only remove the nodes and edges that are not in any
perfect consistent tree \wrt $X$, and (b)~\cmpwpct has the loop invariant that
 at each iteration $i$ (lines~11-15), it always constructs a part of
a perfect tree as a forest. %, via an induction on the iteration .  %, for the nodes with the same topological order of the perfect tree.
%After all the levels are processed $T$

\etitle{Complexity}. The algorithm~\cmpwpct is in time $O(|V||E|+|X|^2+t_{max}*{\cal A})$, where
$t_{max}$ is the maximum time step in $X$, and ${\cal A}$ is the time complexity
of procedure~\lsearch. Indeed, (a) the initialization
and preprocessing phase (lines~1-9) takes $O(|V||E|)$ time, (b) the
sorting phase is in $O(|X|^2)$ time, (c) the bottom-up construction
is in $O(|t_{max}*{\cal A}|)$, which is further bounded by $O(|t_{max}*|V|^3)$ if
an approximable algorithm is used~\cite{RobinsZ05}. 
In our experimental study, we utilize efficient linear programming to compute the
{\em optimal} steiner forest.% for~\lsearch. %, as verified in Section~\ref{experiment}.

The algorithm~\cmpwpct can easily be adapted to the problem of finding the minimum perfect consistent trees,
where each edge has a unit weight.

%\subsection
\stitle{Perfect consistent SP trees}.
 %The computational hardness of %cascade problems based on
 %general independent cascade model is a main drawback
 %in its practical application for real life graphs, as observed in~\cite{Chen10,kimura2006tractable}.
 The independent cascade model may
 be an overkill for real-life applications, as observed in~\cite{Chen10,kimura2006tractable}.
Instead, one may identify the consistent trees which follow
 the shortest path model~\cite{kimura2006tractable}, where cascades propagate
 following the shortest paths.  % as a special case.
 %
 % of the
 %independent cascade model,
 %where the information propagates via shortest paths in social networks.
% Indeed,
%
% (2) as shown in Section~\ref{experiment}, $77\%$ (resp. near $50\%$) of the real cascades of depth no less than $3$ (resp. of depeth $4$) are \sp trees
% in our Twitter cascades.
 % and in average $80\%$ of the paths in cascades of depth as large as $7$ are
 % shortest paths.
%
% cascade following shortest paths.
%Despite that the \mpct and \mwpct problems are \NP-complete and \APX-hard,
 We define a {\em perfect shortest path (\sp) tree} rooted at a given
 source node $s$ as a
perfect consistent tree, such that for each observation point $(v,t) \in X$
of the tree, $t$ = $\dist(s,v)$; in other words,
the path from $s$ to $v$ in the tree is the {\em shortest path} in $G$.
The \mwpct (resp. \mpct) problem for \sp trees is to identify the \sp trees
with the maximum likelihood (resp. minimum size).
%%Based on the perfect \sp trees,
%(a) the \pct problem is to determine if a perfect \sp tree exists in $G$ for $X$;
%the \mpct problem for \sp trees is to find the minimum perfect
%\sp trees \wrt a given partial observation $X$ over graph $G$;
%and (c) similarly,

%\stitle{Remarks}.

 %In contrast to Proposition~\ref{pct} and ~\ref{mwpct},
% We present the following results for \pct,
% \mpct and \mwpct problems based on perfect \sp trees.
% \vspace{-1ex}
 \begin{proposition}
\label{mwpctsp}
Given a graph $G$ and a partial observation $X$,
(a) it is in \PTIME to find a \sp tree \wrt $X$;
(b) the \mpct and \mwpct problems for perfect \sp trees are
\NP-hard and \APX-hard;
%(3) the \mpct problem is approximable within $O(d)$, where $d$ is
%the diameter of the graph $G$;
(c) the \mwpct problem is approximable within $O(d*{\frac{\log f_{min}}{\log f_{max}}})$,
where $d$ is the diameter of $G$, and $f_{max}$ (resp. $f_{min}$) is the maximum (resp. minimum) probability by
the diffusion function $f$.
%they are both approximable
%in polynomial time within $O(|X|^\epsilon)$ for any fixed $\epsilon \leq 0$.
\end{proposition}
 %\vspace{-1ex}

%One may verify Proposition~\ref{mwpctsp} (1) and (2) similarly
%as in the proof of Proposition~\ref{bct} and Proposition~\ref{mbct},
%respectively.

%\etitle{Approximation algorithms}.
We next provide an approximation algorithm to the \mwpct problem for
\sp trees. %, as the proof of Proposition~\ref{mwpctsp} (3) and (4).
%
%The idea is to leverage the solution of the minimum steiner tree problem (\mst).
Given a graph $G$ and a partial observation $X$, the algorithm,
denoted as~\cmpwpctsp (not shown),
first constructs the backbone graph $G_b$  %as the $t_{max}$ \bfs \dag of the source node $s$
as in the algorithm~\cmpwpct. %, where
%$t_{max}$ is the maximum time stamp in $X$.
It then constructs node sets
%$V_r$ and $V_s$, where
$V_r$ = $\{v|(v,t)\in X\}$, and $V_s$ = $V \setminus V_r$.
Treating $V_r$ as required nodes, $V_s$ as steiner nodes, and the log-likelihood function as the weight function,
 ~\cmpwpctsp approximately computes an undirected minimum steiner tree $T$. % rooted at the source node $s$ in $X$.
%This yields an undirected steiner tree $T$.
If the directed counterpart $T'$ of $T$ in $G_b$ is not a tree, ~\cmpwpctsp transforms $T'$ to a tree:
for each node $v$ in $T'$ with more than one parent, it
(a) connects $s$ and $v$ via the shortest path,
and (b) removes the redundant edges attached to $v$.
It then returns $T'$ as an \sp tree.

One may verify that (a) $T'$ is a perfect \sp tree \wrt $X$,
%since there exists a path in the tree $T'$ from source $s$ to any of the node $v$ in the observation point $(v,t)$ with length
%$t$,
 (b) the weight $- L_X(T')$ is bounded by $O(d* \frac{\log f_{min}}{\log f_{max}})$ times of the optimal weight, 
 and (c) the algorithm runs in $O(|V^3|)$ time, leveraging the approximation algorithm for the steiner tree problem~\cite{Vazirani:2001:book}.
\eat{
To see (2), denote
the optimal directed perfect consistent tree as $T^*$.
Observe that (a) the number of edges in the undirected steiner tree $T$ is bounded by $2|T'^*|$, where
$T'^*$ is the optimal undirected steiner tree; (b) at most $d$ extra edges are introduced to a path
from $s$ to a node $v$ in an observation point; and (c) $L_X(T)/L_X(T^*) \leq \frac{|T'| \log f_{min}}{|T^*| \log f_{max}}$
$\leq 2d* \frac{\log f_{min}}{\log f_{max}}$, which is further bounded by $O(d* \frac{\log f_{min}}{\log f_{max}})$.
Moreover, the algorithm runs in $O(|V^3|)$ time, by invoking the approximation algorithm for the steiner tree problem~\cite{Vazirani:2001:book}.
}
\eat{
$T$ is the perfect consistent tree with the minimum weight. Indeed, if there exists another perfect consistent
tree $T'$ with smaller weight, the tree is a steiner tree with smaller weight than $T$, which contradicts our assumption
that $T$ is the optimal.
}
Moreover, the algorithm~\cmpwpctsp can be used for the problem~\mpct for~\sp trees, %to compute the minimum perfect consistent tree,
where each edge in $G$ has the same weight. This achieves an approximation ratio of $d$.  %via a similar analysis
%as remarked earlier. Proposition~\ref{mwpctsp} (3) and (4) thus follows.

%\stitle{Algorithm \cmpmpct}

%\stitle{Algorithm \cmpwpct}

\newcommand{\extc}{\kw{X3C}}

\section{Cascades as bounded trees}
\label{bounded-infer}

In this section, we investigate the cascade inference problems for bounded consistent trees.
\eat{
We first introduce the bounded consistent
tree problem.% in Section~\ref{sec-bct}.
 In Section~\ref{sec-mbct} and Section~\ref{sec-mnwbct}
we introduce two metrics to measure the quality of the bounded consistent tree, as well
as their corresponding optimization problems, respectively.
}
%Given a graph $G$ and a partial observation $X$,
In contrast to the intractable counterpart in Proposition~\ref{pct}(a),
the problem of finding a bounded consistent tree for a given graph and
a partial observation is in \PTIME.

\vspace{1ex}
\begin{proposition}
\label{bct}
For a given graph $G$ and a partial observation $X$,
there is a bounded consistent tree in $G$ \wrt $X$ if and only if
for each $(v,t)\in X$, $\dist(s,v) \leq t$, where $\dist(s,v)$ is
the distance from $s$ to $v$ in $G$.
\end{proposition}
\vspace{1ex}

Indeed, one may verify the following: (a) if there is a node $(v_i,t_i) \in X$ where $\dist(s,v_i) > t_i$, there is no path satisfies
the time constraint and $T$ is empty; (b) if $\dist(s,v_i) \leq t_i$ for each node
$(v_i,t_i) \in X$, a \bfs tree rooted at $s$ with each node $v_i$ in $X$ as its internal node or leaf is a bounded consistent tree.
Thus, to determine whether there is a bounded consistent tree is in $O(|E|)$ time, via a \bfs traversal of $G$ from $s$.

Given a graph $G$ and a partial observation $X$,
the {\em minimum weighted bounded consistent tree} problem, denoted as \mwbct,
is to identify the bounded consistent tree $T^*_s$ \wrt $X$ with the %maximum likelihood $L_X(T^*_s)$.
%Using log-likelihood measurement,
%the \mwbct problem is to identify $T^*_s$ with the 
minimum $- \log L_X(T^*_s)$ (see Section~\ref{sec:preliminary}).
%as remarked earlier in Section~\ref{sec:preliminary}.

%The problem is, nevertheless, nontrivial.

\vspace{1ex}
\begin{theorem}
\label{mwbct}
Given a graph $G$ and a partial observation $X$,
the \mwbct problem is
\begin{itemize}
\vspace{-0.5ex}
\item[(a)]\NP-complete and \APX-hard; and
\vspace{-0.5ex}
\item[(b)] approximable within $O(|X|*{\frac{\log f_{min}}{\log f_{max}}})$,
where $f_{max}$ (resp. $f_{min}$) is the maximum (resp. minimum) probability by
the diffusion function $f$ over $G$.
\end{itemize}
\end{theorem}

We can prove Theorem~\ref{mwbct}(a) as follows. First, the \mwbct problem, as a decision problem, is
to determine whether there exists a bounded consistent tree $T$ with $-L_X(T)$
no greater than a given bound $B$. The problem is obviously in \NP. To show the lower bound, one may show there exists a
polynomial time reduction from the exact 3-cover problem (\extc).
%On the other hand, one may verify
%that the {\em minimum directed steiner tree} () problem is a special case of \mwbct problem,
%which is already hard to approximate~\cite{Vazirani:2001:book}.%not approximable within $O(\log|V|)$~\cite{Vazirani:2001:book}.
Second, to see the approximation hardness, one may verify that there exists an \AFPR from the minimum directed steiner tree (\mst) problem.
%minimum steiner tree problem.
%\begin{proof}
%proof of (1): reduct from X3C;
%proof of (2): reduct from X3C;
%\end{proof}

%Despite of the hardness as shown by (1) of Theorem~\ref{mwbct}, we provide,
%as a proof of  Theorem~\ref{mwbct} (2),
We next provide a polynomial time algorithm, denoted as \cmpwbct, for the \mwbct problem.
The algorithm runs in {\em linear time} \wrt
the size of $G$, and with performance guarantee as in Theorem~\ref{mwbct}(b).

%%%%%%%%%%%%%%%%%%%%%%%%%%%%%%%%%%%%%%%%
\begin{figure}[tb!]
\vspace{-1.5ex}
\begin{center}
\fbox{
{\small
%\fcolorbox{black}{yellow}{
%\fbox{
\begin{minipage}{3.3in}
%\myhrule
%\vspace{-2.5ex}
\mat{0ex}{
%{\bf Algorithm}~\pSim\\
\sstab {\sl Input:\/} \= graph $G$ and partial observation $X$.\\
{\sl Output:\/} a bounded consistent tree $T$ in $G$. \\
\stab \bcc \hspace{2ex} \= tree $T$ = $(V_t,E_t)$, where $V_t$ := $\{s| (s,0)\in X\}$, $E_t$ := $\emptyset$; \\
%\icc \> /*sort X as a list L*/ \\
%\> list $L$ := $\{(s,0),(v_1,t_1),\ldots, (v_k,t_k)\}$ \\
%\> where $t_i \leq t_{i+1}$, $i \in [1,k-1]$; \\
\icc \> compute $t_k$ bounded \bfs \dag $G_d$ of $s$ in $G$; \\
%\icc \> topologically sort $G_d$;
%\icc \> set the level $l(v)$ for each node $v$ in $G_d$ as $\dist(s,v)$; \\
\icc \> \For {\bf each} $t_i \in [t_1, t_k]$ \Do \\
\icc \> \hspace{2ex}  \For  {\bf each} node $v$ where $(v, t_i) \in X$ and $l(v)$ = $i$ \Do \\
\icc \> \hspace{4ex} \If $i > t_i$ \Then \Return $\emptyset$; \\
\icc \> \hspace{4ex} find a path $\rho$ from $s$ to $v$ with the \\
\> \hspace{4ex} minimum weight $w(\rho)$ = $ -\Sigma \log f(e)$ for each $e \in \rho$; \\
%\icc \> \hspace{4ex} \For {\bf each} node $v' \in V_t$ where $l(v')$ = $i - 1$ \Do \\
%\icc \> \hspace{6ex} select $e$ = $(v',v)$ with minimum $\log f(e)$ in $G$; \\
\icc \> \hspace{4ex} $T$ = $T \cup \rho$; \\
\eat{
\icc \> \For {\bf each} $(v_i,t_i) \in X$ \Do \\
\icc \> \hspace{2ex}\= \If $\dist(s,v_i) > t_i$ \Then \Return $\emptyset$; \\
\icc \> \> path $\rho$ := $\emptyset$ with weight $w(\rho)$ := $\infty$; \\
\icc \> \> \For \= {\bf each} $v \in V_t$ \Do \\
\icc \> \> \> find the path $\rho'$ from $v$ to $v_i$ with the \\
\>\>\> minimum weight $w(\rho')$ = $ \Sigma \log f(e)$ for each $e \in \rho'$; \\
\icc\>\> \If $w(\rho) \geq w(\rho')$ \Then
\icc \> \>\> $T$ = $T \cup \rho$; \\
}
\icc \> \Return $T$ as a bounded consistent tree; \\
}
\vspace{-5ex}
\myhrule
\end{minipage}
}
}
\end{center}
\vspace{-2ex}
%\vspace{-2.5ex}
\caption{Algorithm~\cmpwbct: searching bounded consistent trees via top-down strategy}
\label{fig-cmpwbct}
\vspace{-2ex}
%\vspace{-2ex}
\end{figure}
%%%%%%%%%%%%%%%%%%%%%%%%%%%%%%%%%%%%%

\stitle{Algorithm}. The algorithm~\cmpwbct is illustrated in Fig.~\ref{fig-cmpwbct}.
Given a graph $G$ and a partial observation $X$, the algorithm first initializes
a tree $T$ = $(V_t,E_t)$ with the single source node $s$ (line~1).
It then computes the {\em $t_k$} bounded \bfs directed acyclic graph (\DAG)~\cite{Gutin08} $G_d$
 of the source node $s$, where $t_k$ is
the maximum time step of the observation points in $X$, and $G_d$ is a
\DAG~induced by the nodes and edges visited by a \bfs traversal of $G$ from $s$ (line~2).
%It then sets the topological order of a node $v$ in $G_d$ as a value $l(v)$ (line~3).
Following a top-down strategy, for each node $v$ of $(v,t) \in X$, ~\cmpwbct then (a) selects a path $\rho$ with the minimum
$\Sigma \log f(e)$ from $s$ to $v$,
 and (b) extends the current tree $T$ with the path $\rho$ (lines~3-7).  If for some observation point $(v,t) \in T$,
$\dist(s,v) > t$, then~\cmpwbct returns $\emptyset$ as the tree $T$ (line~5).
Otherwise, the tree $T$ is returned (line~8) after all the observation points in $X$
are processed.

\etitle{Correctness and complexity}.
One may verify that algorithm~\cmpwbct either correctly computes a bounded consistent tree $T$, or
returns $\emptyset$. For each node in the observation point $X$,
there is a path of weight selected using a greedy strategy, and the top-down strategy
guarantees that the paths form a consistent tree. %Using
The algorithm runs in time $O(|E|)$, since it visits
each edges at most once following a \bfs traversal.

We next show the approximation ratio in Theorem~\ref{mwbct}(b).
% that the total weight $- L_X(T)$ of the tree $T$
% is bounded by $O(|X|*{\frac{\log f_{min}}{\log f_{max}}})$
%times of the optimal solution.
Observe that for a single node $v$ in $X$, (a) the total weight of the
path $w$ from $s$ to $v$ is no greater than $-|w|\log f_{min}$, where
$|w|$ is the length of $w$; and (b) the weight of the counterpart of $w$ in $T^*$,
denoted as $w'$, is no less than $-|w^*|\log f_{max}$.
Also observe that $|w| \leq |w^*|$. Thus, $w/w^* \leq \frac{\log f_{min}}{\log f_{max}}$.
As there are in total $|X|$ such nodes, $L_X(T)/L_X(T^*) \leq |X| \frac{w}{w^*} \leq |X|\frac{\log f_{min}}{\log f_{max}}$. Theorem~\ref{mwbct}(b) thus follows.

%dynamic programming, the algorithm runs in time $O(|E|)$.
%One may also verify that

\stitle{Minimum bounded consistent tree}.
%\label{sec-mbct}
We have considered the likelihood function as a quantitative metric %minimum bounded consistent trees as a minimum inferences for cascades.
for the quality of the bounded consistent trees.
As remarked earlier, one may simply want to identify the bounded consistent trees of the {\em minimum} size.
%We introduce the minimum bounded consistent tree problem.
Given a social graph $G$ and a partial observation $X$,
the {\em minimum bounded consistent tree problem}, denoted as \mbct, is to identify the
bounded consistent tree with the minimum size, \ie the total number of nodes and edges.
The \mbct problem is a special case of \mwbct,
and its main result is summarized as follows.

\begin{proposition}
\label{mbct}
The \mbct problem is (a) \NP-complete,
(b) \APX-hard,
and (c) approximable within $O(|X|)$, where
$|X|$ is the size of the partial observation $X$. % add time complexity
%time $O(|V||E|)$.
\end{proposition}

%\textbf{proof sketch}:
%\begin{proof}
Proposition~\ref{mbct}(a) and~\ref{mbct}(b) can both be shown by constructing reductions
from the \mst problem, which is \NP-complete and \apx-complete~\cite{Vazirani:2001:book}.
%we consider a decision problem
%which is to determines whether there exists a bounded consistent tree
%with size smaller than a given bound $B$. Observe that
%(a) \mbct is in \NP; (b) to show the lower bound,
%we may construct a reduction from
%the {\em minimum steiner tree} problem (\mst).

%Proposition~\ref{mbct}(1) shows that the problem is nontrivial. One may
%turn to identify polynomial time approximation algorithm for \mbct. Nevertheless,
%the problem is hard to approximate.
%To see Proposition~\ref{mbct}(2), we may conduct an
%approximation ration preserving reduction from \mst as an
%optimization problem. It is known that \mst is \apx-complete~\cite{Vazirani:2001:book}.
%Thus, \mbct is \apx-hard.
%proof of (1): reduct from minimum steiner tree;
%proof of (2): reduct from minimum (weighted) steiner tree;
%\end{proof}
%\eof
%
%\stitle{Algorithm}.
%Proposition~\ref{mbct}(2) indicates that it is not likely to find a polynomial time
%algorithm which approximates the problem within any given constant ratio.
Despite of the hardness, the problem can be approximated within $O(|X|)$ in polynomial time,
by applying the algorithm~\cmpwbct over an instance where each edge has a unit weight. 
This completes the proof of Proposition~\ref{mbct}(c).

\newcommand{\lpbct}{\kw{BCT_{lp}}}
\newcommand{\lppct}{\kw{PCT_{lp}}}
\newcommand{\rbct}{\kw{BCT_r}}
\newcommand{\rpct}{\kw{PCT_r}}
\newcommand{\gpct}{\kw{PCT_{g}}}
\newcommand{\precision}{\kw{prec}}
\newcommand{\recall}{\kw{rec}}

\section{Experiments}
\label{experiment}

We next present an experimental study of our proposed methods.
Using both real-life and synthetic data, we conduct
three sets of experiments to evaluate
(a) the effectiveness of the proposed algorithms,
(b) the efficiency and the scalability of~\cmpwpct and~\cmpwbct.

\begin{figure*}[tb!]
\vspace*{-3ex}
\begin{center}
%PCT & BCT @ Enron%
\subfigure[PCT@Enron:~\precision]{\label{fig-pct-enron-prec} %fig-pct-enron-prec
\includegraphics[width=4.1cm,height=3.0cm]{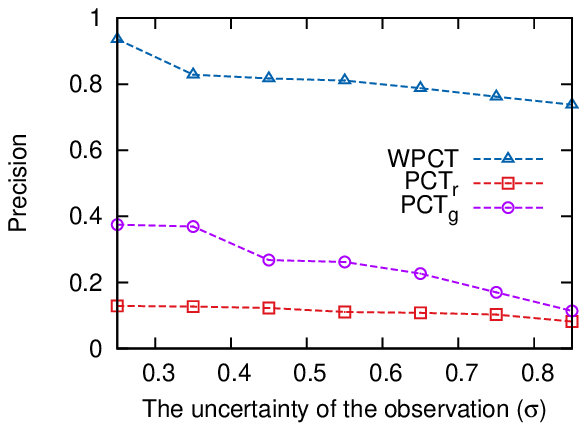}}
\hfill
\subfigure[PCT@Enron:~\recall]{\label{fig-pct-enron-rec}
\includegraphics[width=4.1cm,height=3.0cm]{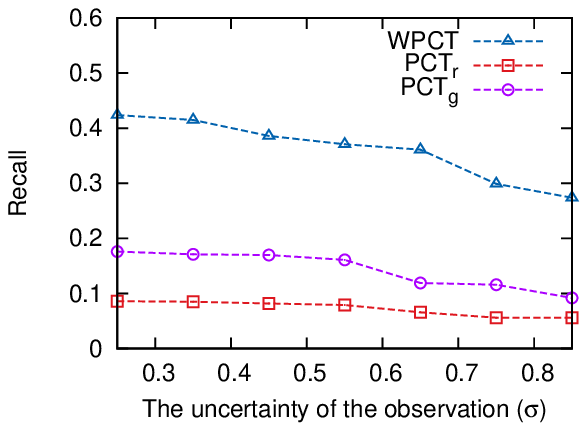}}
\hfill
\subfigure[BCT@Enron:~\precision]{\label{fig-bct-enron-prec} %fig-pct-enron-prec
\includegraphics[width=4.1cm,height=3.0cm]{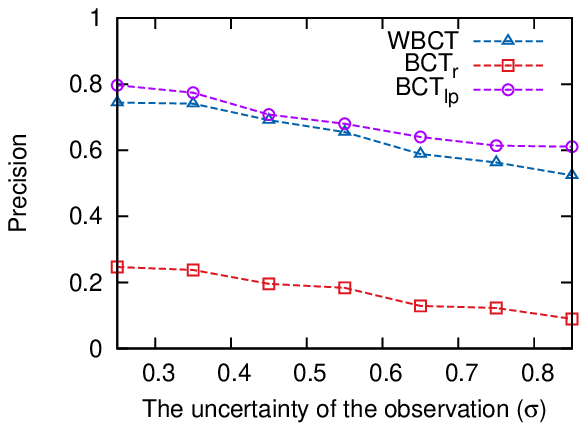}}
\hfill
\subfigure[BCT@Enron:~\recall]{\label{fig-bct-enron-rec}
\includegraphics[width=4.1cm,height=3.0cm]{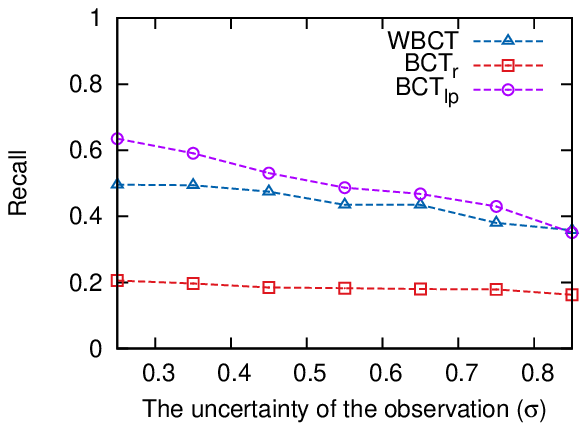}}
\hfill
%PCT & BCT @ Retweet Cascades%
\subfigure[PCT@RT:~\precision]{\label{fig-pct-rt-prec} %fig-pct-enron-prec
\includegraphics[width=4.1cm,height=3.0cm]{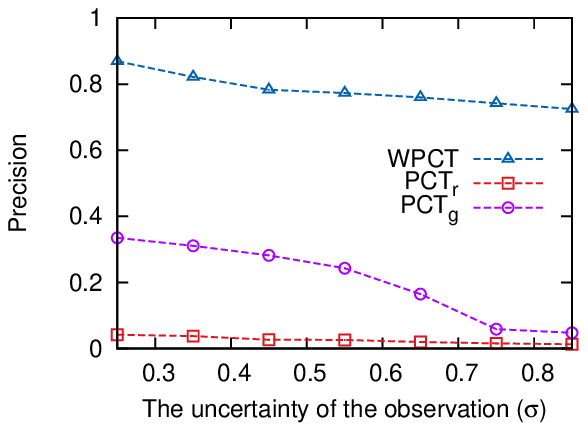}}
\hfill
\subfigure[PCT@RT:~\recall]{\label{fig-pct-rt-rec}
\includegraphics[width=4.1cm,height=3.0cm]{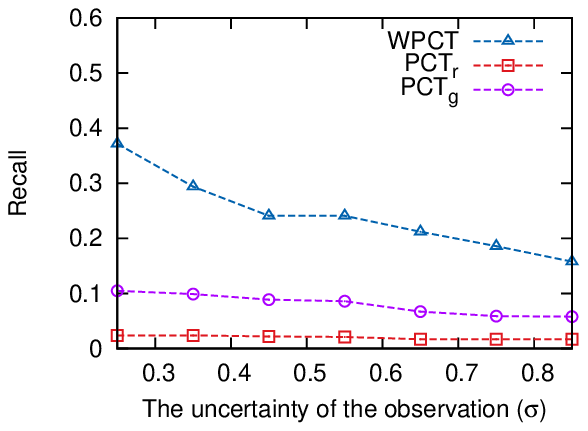}}
\hfill
\subfigure[BCT@RT:~\precision]{\label{fig-bct-rt-prec} %fig-pct-enron-prec
\includegraphics[width=4.1cm,height=3.0cm]{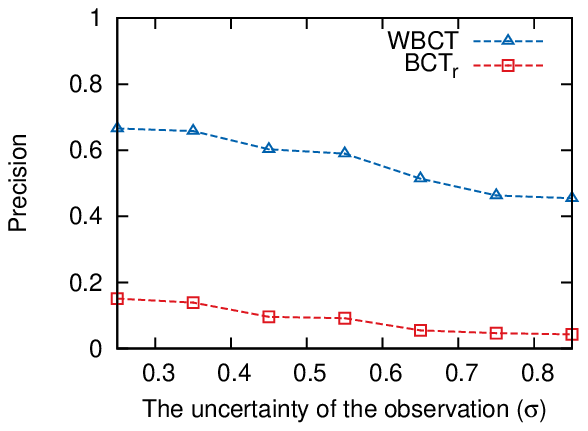}}
\hfill
\subfigure[BCT@RT:~\recall]{\label{fig-bct-rt-rec}
\includegraphics[width=4.1cm,height=3.0cm]{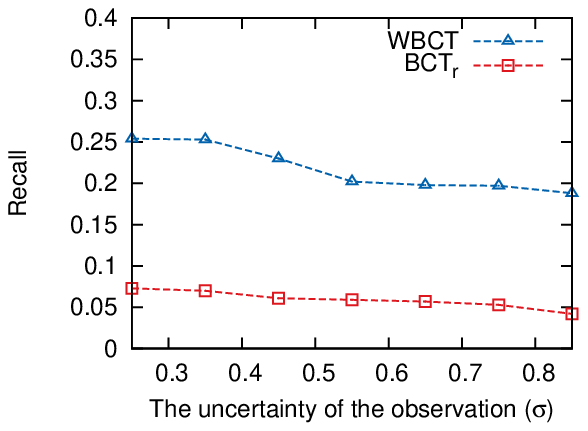}}
\hfill
%\subfigure[Precision over Twitter]{\label{fig-ptwitterprec}
%\includegraphics[width=4.1cm,height=3.0cm]{./figure/retweet4_perfect_precision.eps}}
%\hfill
%\subfigure[Recall over Twitter]{\label{fig-btwitterprec}
%\includegraphics[width=4.1cm,height=3.0cm]{./figure/retweet4_perfect_recall.eps}}
%\hfill
%
%\subfigure[Precision over Enron]{\label{fig-benronprec}
%\includegraphics[width=4.1cm,height=3.0cm]{./figure/enron4_bounded_precision.eps}}
%\hfill
%\subfigure[Recall over Enron]{\label{fig-benronrec}
%\includegraphics[width=4.1cm,height=3.0cm]{./figure/enron4_bounded_recall.eps}}
%\hfill
%\vspace*{-3ex}
% Exp 2: efficiency of bounded tree model
%%\subfigure[Varying $|T|$]{\label{fig-bcsize}
%%\includegraphics[width=4.1cm,height=3.0cm]{./figure/cascade_size_bounded}}
%%\hfill
%%\eat{\subfigure[Varying $d$]{\label{fig-bcdepth}%Efficiency of \cmpwbct: Twitter
%%\includegraphics[width=4.1cm,height=3.0cm]{./figure/depth_bounded.eps}}
%%\hfill}
%%\subfigure[Varying $\sigma$]{\label{fig-bunc}
%%\includegraphics[width=4.1cm,height=3.0cm]{./figure/unc_bounded.eps}}
%\vspace*{-3ex}
% Exp 3: efficiency of perfect tree model
%\hfill
%%\subfigure[Varying $|T|$]{\label{fig-pcsize}
%%\includegraphics[width=4.1cm,height=3.0cm]{./figure/cascade_size_perfect.eps}}
%%\hfill
%%\eat{\subfigure[Varying $d$]{\label{fig-pcdepth}%time delay
%%\includegraphics[width=4.1cm,height=3.0cm]{./figure/depth_perfect.eps}}
%%\hfill}
%%\subfigure[Varying $\sigma$]{\label{fig-punc}
%%\includegraphics[width=4.1cm,height=3.0cm]{./figure/unc_perfect.eps}}
%%\eat{}
\end{center}
\vspace*{-2ex}
%\caption{Accuracy over real-life cascades}
\caption{The \precision and \recall of the inference algorithms over Enron email cascades and Retweet cascades}
\label{fig-perf} %Experimental Results: the effectiveness
\vspace{-3ex}
\end{figure*}

\stitle{Experimental setting}. We used real-life data to
evaluate the effectiveness of our methods, and
synthetic data to conduct an in-depth analysis on scalability by varying the parameters of cascades and partial observations.

\etitle{(a) Real-life graphs and cascades}. We used the following real-life datasets.
%including both real-life interaction network graphs and
%cascades from emperical study, where all the cascades are trees following
%the independent cascade model.
(i) {\em Enron email cascades}. The dataset of \emph{Enron Emails}~\footnote{http://www.cs.cmu.edu/~enron/}
consists of a social graph of $86,808$ nodes and $660,642$ edges, where
a node is a user, and two nodes are connected if there is an email message between them.
We tracked the \emph{forwarded} messages of the same subjects and
obtained $260$ cascades of depth no less than $3$ with more than $8$ nodes.
(ii) {\em Retweet cascades} (RT). The dataset of \emph{Twitter Tweets}~\footnote{http://snap.stanford.edu/data/twitter7.html}~\cite{WSDM:Yang:2011a}
contains more than $470$ million posts from more than $17$ million users, %and $134.8$ million edges,
covering a period of 7 months from June 2009. We extracted
the retweet cascades of the identified {\em hashtags}~\cite{WSDM:Yang:2011a}. To guarantee that a cascade represents the propagation of a single hashtag, we removed those retweet cascades containing multiple hashtags. In the end, we
obtain $321$ cascades of depth more than $4$, with node size ranging from $10$ to $81$.
%To estimate the diffusion function for the network graph,
Moreover, we used the EM algorithm from~\cite{Saito:2008a:KES08} to estimate the diffusion function.

\etitle{(b) Synthetic cascades}.  We generated a set of synthetic cascades unfolding in an anonymous Facebook social graph~\footnote{http://current.cs.ucsb.edu/socialnets}, which exhibits properties such as power-law degree distribution, high clustering coefficient and positive assortativity~\cite{Wilson:2009:EuroSys09}. The diffusion function is constructed by randomly assigning real numbers between $0$ and $1$ to edges in the network. The generating process is controlled by size $|T|$. We randomly choose a node as the source of the cascade. By simulating the diffusion process following the independent cascade model, we then generated cascades \wrt $|T|$ and assigned time steps. % as well.
%\etitle{(2) Synthetic cascades}.

\etitle{(c) Partial observation.} For both real life and synthetic cascades,
we define {\em uncertainty} of a cascade $T$ as $\sigma$ = $1- \frac{|X|}{|V_T|}$,
where $|V_T|$ is the size of the nodes in $T$, and $|X|$ is the size of the partial
observation $X$. We remove the nodes from the given cascades
until the uncertainty is satisfied, and collect the remaining nodes and their time steps as $X$.

\etitle{(d) Implementation}.
We have implemented the following in C++:
(i) algorithms~\cmpwpct, and~\cmpwbct;
(ii) two linear programming algorithms~\lppct and~\lpbct,
which identify the optimal weighted bounded consistent trees and the optimal
perfect consistent trees using linear programming, respectively;
(iii) two randomized algorithms~\rpct and~\rbct, which are developed to
randomly choose trees from given graphs.~\rpct is developed using
a similar strategy for~\cmpwpct, especially for each level the steiner forest is
randomly selected (see Section~\ref{perfect-infer}); as~\cmpwbct does,~\rbct runs on bounded~\bfs directed acyclic graphs, but randomly selects edges.
(iv) to verify various implementations of~\cmpwpct, an algorithm~\gpct is
developed by using a greedy strategy to choose the steiner forest for each level
(see Section~\ref{perfect-infer}).
We used LP\_solve 5.5~\footnote{http://lpsolve.sourceforge.net/5.5/} as the linear programming solver.
%We also used the C++ BOOST library for the randomness. % and random generator.

We used a machine powered by an {\small Intel(R)} Core
 $2.8$GHz CPU and $8$GB of RAM, using
Ubuntu 10.10.  Each experiment was run by $10$ times and
 the average is reported here.

%\stitle{Exp-2: Efficiency over real life cascades}.
\begin{figure*}[tb!]
\begin{center}
\subfigure[Varying $|T|$]{\label{fig-bcsize}
\includegraphics[width=4.1cm,height=3.0cm]{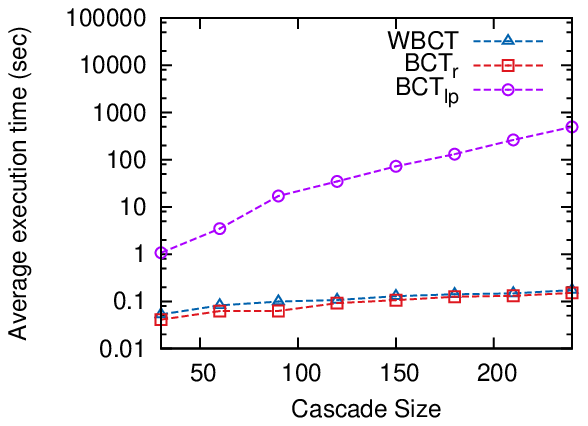}}
\hfill
\eat{\subfigure[Varying $d$]{\label{fig-bcdepth}%Efficiency of \cmpwbct: Twitter
\includegraphics[width=4.1cm,height=3.0cm]{depth_bounded.eps}}
\hfill}
\subfigure[Varying $\sigma$]{\label{fig-bunc}
\includegraphics[width=4.1cm,height=3.0cm]{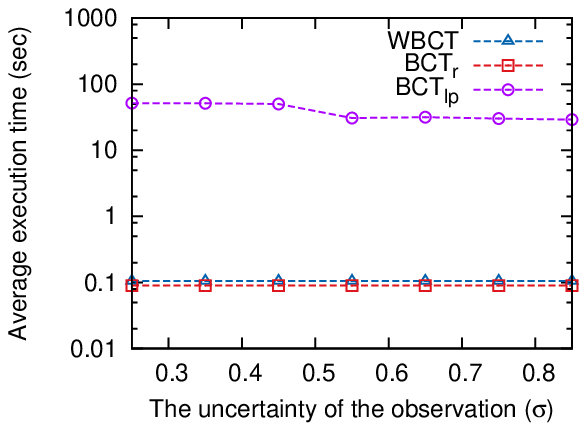}}
%\vspace*{-3ex}
% Exp 3: efficiency of perfect tree model
%\hfill
\subfigure[Varying $|T|$]{\label{fig-pcsize}
\includegraphics[width=4.1cm,height=3.0cm]{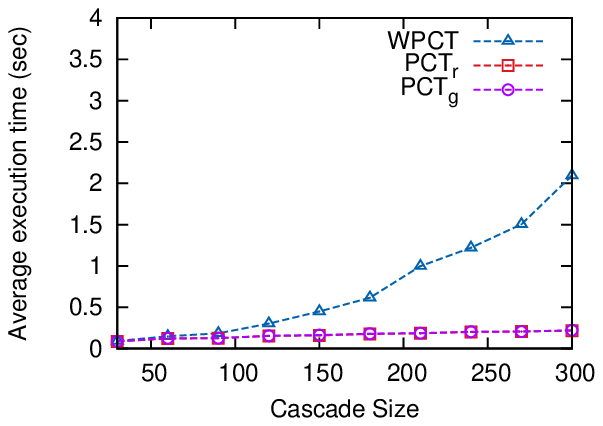}}
\hfill
\eat{\subfigure[Varying $d$]{\label{fig-pcdepth}%time delay
\includegraphics[width=4.1cm,height=3.0cm]{depth_perfect.eps}}
\hfill}
\subfigure[Varying $\sigma$]{\label{fig-punc}
\includegraphics[width=4.1cm,height=3.0cm]{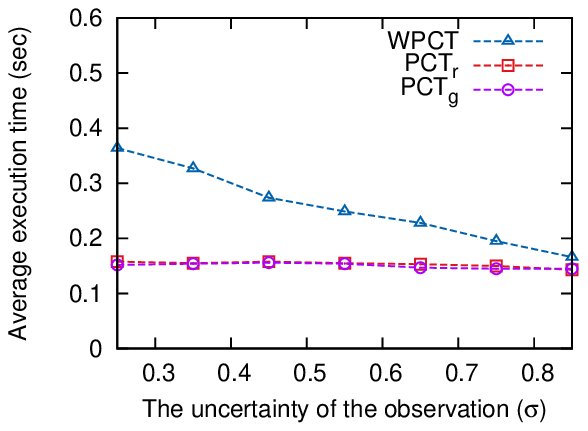}}
\end{center}
\vspace{-2ex}
\caption{Efficiency and scalability over synthetic cascades} \label{fig-efficiency}
\vspace{-3ex}
\end{figure*}

%\textbf{Precision}
\begin{table}[tb!]
\vspace{-1ex}
\begin{small}
%\scriptsize
\begin{center}
\begin{tabular}{|r|c||c|c|c|c|}
\cline{1-6} \multicolumn{1}{|c||}{} & \multicolumn{1}{c||}
{} &
\multicolumn{2}{c|}{\textbf{Enron}} &
\multicolumn{2}{c|} {\textbf{Twitter}}\\
\cline{3-6} \multicolumn{1}{|c||}{\raisebox{1.5ex}[0pt]{\textbf{Algorithms}}}
& \multicolumn{1}{|c||}{\raisebox{1.5ex}[0pt]{\textbf{Precision}}}
& $d$=$3$ & $d$=$4$ & $d$ = $4$ & $d$ = $5$ \\

\hline\hline \multicolumn{1}{|c||}{} & \multicolumn{1}{c||}{$\precision_v$}
& $100 \%$ & $100\%$ & $97.2\%$ & $93.2\%$ \\

\cline{2-6} \multicolumn{1}{|c||}{\raisebox{1.5ex}[0pt]{\textbf{\cmpwpct}}} &
\multicolumn{1}{c||} {$\precision_e$} & $78.2\%$ & $82.4\%$ & $86.1\%$ & $82.6\%$ \\
\hline

\hline\hline \multicolumn{1}{|c||}{} & \multicolumn{1}{c||}{$\precision_v$}
& $100 \%$ & $70.1\%$ & $73.6\%$ & $66.1\%$ \\

\cline{2-6} \multicolumn{1}{|c||}{\raisebox{1.5ex}[0pt]{\textbf{\cmpwbct}}} &
\multicolumn{1}{c||} {$\precision_e$} & $69\%$ & $55.7\%$ & $60.6\%$ & $41.7\%$ \\
\hline
\end{tabular}
\end{center}
\end{small}
\vspace{-1ex}
\caption{$\precision_v$ and $\precision_e$ over real cascades}
\label{tab-realdata}
\vspace{-5ex}
\end{table}

\stitle{\bf Experimental results}. We next present our findings.

\vspace{-0.8ex}
\etitle{Effectiveness of consistent trees}.
In the first set of experiments, using real life cascades,
we investigated the accuracy and the efficiency of our cascade inference algorithms.

\stab
(a) Given a set of real life cascade $\textbf{T}$ = $\{T_1, \ldots, T_k\}$,
for each cascade $T_i = (V_{T_i},E_{T_i}) \in \textbf{T}$,
we computed an inferred cascade ${T_i}'$ = $(V_{{T_i}'},E_{{T_i}'})$ according to
a partial observation with uncertainty $\sigma$. Denote the nodes in
the partial observation as $V_X$. We evaluated the
{\em precision} as $\precision$ = $\frac{\Sigma (|(V_{{T_i}'}\cap V_{T_i}) \setminus V_X|)}{\Sigma (|V_{{T_i}'} \setminus V_X)|}$,
and $\recall$ = $\frac{\Sigma (|(V_{{T_i}'}\cap V_{T_i}) \setminus V_X|)}{\Sigma (|V_{T_i} \setminus V_X)|}$.
Intuitively,~\precision is the fraction of inferred nodes that are missing from $T_i$, while~\recall is the fraction of missing nodes that are inferred by ${T_i}'$.
%
%and \recall
%are the ratios of the nodes our algorithms infer
%measure the ``recovering'' ability of our inference algorithms for
%the unknown nodes of the original cascades, based on $G$ and $X$.

%Varying $\sigma$ from $0.25$ to $0.8$,
For Enron email cascades, Fig.~\ref{fig-pct-enron-prec} and Fig.~\ref{fig-pct-enron-rec}
show the accuracy of~\cmpwpct,~\gpct and~\rpct for inferring cascades, while
$\sigma$ is varied from $0.25$ to $0.85$.~\lppct does not scale over the Enron dataset and thus is not shown. (i)~\cmpwpct
outperforms~\gpct and~\rpct on both~\precision and~\recall.
%When $25\%$ of the nodes are unknown, ~\cmpwpct
%achieves precision as high as $72\%$.
%$70\%$ of them.
(ii) When the uncertainty increases, both the~\precision and~\recall of the three algorithms decrease.
In particular,~\cmpwpct successfully infers cascade nodes with~\precision
no less than $70\%$ and~\recall no less than $25\%$ even when $85\%$ of the nodes in the
cascades are removed.
%~\cmpwpct successfully infers more than $50\%$ of
%the unknown nodes. %, from only $25\%$ observed nodes.
Using the same setting, the performance of~\cmpwbct,~\lpbct and~\rbct are shown in Fig.~\ref{fig-bct-enron-prec} and Fig.
~\ref{fig-bct-enron-rec}, respectively. (i) Both~\lpbct and~\cmpwbct outperform~\rbct, and their~\precision and~\recall decrease while the uncertainty increases. (ii)~\lpbct has better performance than~\cmpwbct. In particular, both~\lpbct and~\cmpwbct successfully infer the cascade nodes with the~\precision no less than $50\%$ and with the~\recall no less than $25\%$, even when $85\%$ of the nodes in the cascades are removed.

%We summerize the accuracy of~\cmpwpct over Twitter cascades below.
For retweet cascades, the~\precision and the~\recall of~\cmpwpct,~\gpct and~\rpct are
%in general lower than their %counterpart over Enron cascades,
shown in Fig.~\ref{fig-pct-rt-prec} and in Fig.~\ref{fig-pct-rt-rec}, respectively.
While the uncertainty increases from $0.25$ to $0.85$, (i)~\cmpwpct outperform~\rpct and~\gpct,
and (ii) %both the~\precision and the~\recall
the performance of all the algorithms decreases.
In particular,~\cmpwpct successfully infers the nodes with the~\precision more than $80\%$ and the~\recall more than $35\%$,
while the uncertainty is $25\%$.
Similarly, the~\precision and the~\recall of~\cmpwbct and~\rbct are presented in Fig.~\ref{fig-bct-rt-prec} and Fig.~\ref{fig-bct-rt-rec}, respectively.
As~\lpbct does not scale on retweet cascades, its performance is not shown.
While the uncertainty $\sigma$ increases, the~\precision and the~\recall of the algorithms decrease.
For all $\sigma$,~\cmpwbct outperforms~\rbct;
in particular,~\cmpwbct correctly infers the nodes with~\precision no less than $60\%$ and~\recall no less than $25\%$,
when $\sigma$ is $25\%$.

%not as good, % as their counterparts over Enron cascades,
% due to the more complex topological structure
% of the Twitter network. Specifically,
% the~\precision and~\recall of~\cmpwpct ranges from
% $50\%$ to $30\%$, and $20\%$ to $10\%$, respectively,
% while ~\gpct and~\rpct infer less than $1\%$ of
% the unknown nodes in all cases (not shown).

 %We also test the precision and recall of~\cmpwbct comparing with~\rbct and~\lpbct.
 %The \precision and \recall of all three algorithms are not sensitive to
 %the change of~$\sigma$. In both datasets, \precision and \recall of~\cmpwbct average around
 %$20\%$ and $21\%$, respectively,
 %while \precision and \recall of~\rbct are always less than $10\%$. The \precision and \recall of~\lpbct
 %are close to those of~\cmpwbct.

\stab
(b) To further evaluate the structural similarity of $T_i$ and ${T_i}'$ as described in (a), % of~\cmpwpct and~\cmpwbct,
we also evaluate
(i) {\em $\precision_v$} = $\frac{|V''|}{|V'|}$ for nodes $V'$ = $(V_{{T_i}'}\cap V_{T_i}) \setminus V_X$,
where $V''\in V'$ are the nodes with the same {\em topological order} in both $T_i'$ and $T_i$,
and
(ii) {\em $\precision_e$} = $\frac{|E'|}{|E_{{T_i}'}|}$ for $E'$ = $E_{T_i}\cap E_{{T_i}'}$, following
the metric for measuring graph similarity~\cite{Raymond02rascal}.
The average results
are as shown in Table~\ref{tab-realdata}, for $\sigma$ =$50\%$, and the cascades
of fixed depth. As shown in the table,
\eat{
The algorithm~\cmpwpct, over both datasets,
infers cascades with satisfiable structural
similarity \wrt the original cascades:}
for~\cmpwpct, the average $\precision_v$ is above $90\%$,
and the average $\precision_e$ is above $75\%$ over both datasets.
Better still, the results hold even when we set $\sigma$ = $85\%$.
For~\cmpwbct, $\precision_v$ and $\precision_e$ are above $65\%$ and above $40\%$, respectively.
For~\cmpwpct, $\precision_v$ and $\precision_e$ have almost consistent performance on both datasets; however, for~\cmpwbct, the $\precision_v$ and $\precision_e$ of the inferred Enron cascades are higher than those of the inferred retweet cascades. The gap might result from the different diffusion patterns between these two datasets: we observed that there are more than $70\%$ of cascades in the Enron dataset whose structures are contained in the \bfs directed acyclic graphs of~\cmpwbct, while in the Twitter Tweets there are less than $45\%$ of retweet cascades following the assumed graph structures of~\cmpwbct.
%On the other hand, in general, ~\cmpwpct achieves better precision
%than~\cmpwbct, which shows that perfect consistent trees
%are more suitable for accurate observations.
%which is more strict than bounded consistent trees.
%as remarked earlier.

\etitle{Efficiency over real datasets}.
In all the tests over real datasets, ~\rpct,~\rbct,~\gpct and~\cmpwbct take
less than $1$ second.~\lpbct does not scale for retweet cascades, while ~\lppct does not scale for both datasets.
On the other hand, while~\cmpwpct takes less than $0.4$ seconds in inferring
all the Enron cascades, it takes less than $20$ seconds to infer Twitter cascades where
$d$=$4$, and $100$ seconds when $d$ = $5$. Indeed, for Twitter network the average
degree of the nodes is $20$, %as observed in~\cite{Kwak10www},
while the average degree for Enron dataset is $7$. As such, it takes more time for~\cmpwpct
to infer Twitter cascades in the denser Twitter network. In our tests, the efficiency of all the algorithms
are not sensitive \wrt the changes to $\sigma$.

\etitle{Efficiency and scalability over synthetic datasets}.
In the second set of experiments, we evaluated the efficiency and the scalability of our algorithms using synthetic cascades.
%First, we evaluated the performance of~\cmpwpct compared with~\rpct and~\gpct. Second, we evaluated the performance of~\cmpwbct compared with~\lpbct and~\rbct.

%\stab
%(a) \etitle{Efficiency and scalability of~\cmpwpct}.
\stab
(a) We first evaluate the efficiency and scalability of~\cmpwpct and compare~\cmpwpct with~\rpct and~\gpct.

Fixing uncertainty $\sigma = 50\%$, we varied $|T|$ from $30$ to $240$. Fig.~\ref{fig-pcsize} shows that~\cmpwpct scales well with
the size of the cascade. Indeed, it only takes $2$ seconds to infer the cascades with $300$ nodes.

Fixing size $|T|=100$, we varied the uncertainty $\sigma$ from $0.25$ to $0.85$. Fig.~\ref{fig-punc} illustrates that while all the three algorithms are more efficient with larger $\sigma$,~\cmpwpct is more sensitive. All the three algorithms scale well with $\sigma$.

As~\lppct does not scale well, its performance is not shown in Fig.~\ref{fig-pcsize} and Fig.~\ref{fig-punc}.

%Fixing $d$ = $4$ and $\sigma$ = $50\%$, we varied
%the size of the cascades to be inferred from $30$ to $300$.

%On the other hand, the efficiency of~\cmpwpct is comparable with
%~\gpct and~\rpct, which in turn have low accuracy. {\em may add more}.

%Figure~\ref{fig-pcdepth} shows that all three algorithms scale well.
%with the cascade depth varied from $3$ to $7$, while $\sigma$ = $50\%$ %.
%(1) For cascades with depth
%$7$ and $100$ nodes, ~\cmpwpct finds inferred cascades as perfect consistent
%trees in less than $5$ seconds. (2)
%Notably, ~\cmpwpct is more sensitive than~\cmpwbct
%\wrt $d$ (see Fig.~\ref{fig-bcdepth}), since \cmpwpct identifies trees
%with exact path lengths satisfying the time step constraints, while
%~\cmpwbct only finds shortest paths.% via \bfs traversal.
%~ from the source node to observed nodes following a bottom-up
%strategy, and at each level~\cmpwpct computes the local optimum solution.

%\stab
%(b) \etitle{Efficiency and scalability of~\cmpwbct}.
\stab
(b) Using the same setting, we evaluated the performance of~\cmpwbct,
compared with~\lpbct and~\rbct.
%Varying cascade size.
%\stab

Fixing $\sigma$ and varying $|T|$, the result is reported in Fig.~\ref{fig-bcsize}.
First,~\cmpwbct outperforms~\lpbct, and is almost
as efficient as the randomized algorithm~\rbct. For the cascade of
$240$ nodes,~\cmpwbct takes less than $0.5$ second to infer the structure,
while~\lpbct takes nearly $1000$ seconds. Second, while~\cmpwbct is not sensitive to the change of $|T|$,~\lpbct is much more sensitive.

%(2) %Varying cascade depth.
%To evaluate the impact of cascade depth
%sizes on the scalability, we fixed $\sigma$ = $50\%$ and
%$|C|$ = $100$, while varying $d$ from $3$ to $7$.
%Figure~\ref{fig-bcdepth} shows that all three algorithms
%are sensitive to the change of $d$. This is because
%with the increasement of $d$, all three algorithms need to
%investigate longer paths \wrt the larger time steps in $|X|$.
%On the other hand, ~\cmpwbct is much more efficient than ~\lpbct, verifying the
%result in Fig~\ref{fig-bcsize}.

%Varying $|X|$.
%\stab
Fixing $|T|$ and varying $\sigma$, Fig.~\ref{fig-bunc} shows the performance of the three algorithms. The figure tells us that~\cmpwbct and~\rbct are
 less sensitive to the change of $\sigma$ than~\lpbct. This is because~\cmpwbct and~\rbct identify
 bounded consistent tree by constructing shortest paths from the source to the observed nodes.
 %via dynamic programming.
 When the maximum depth of the observation point is fixed, the total number of nodes and edges visited
 by~\cmpwbct and~\rbct are not sensitive to $\sigma$.

%We also verify the performance of~\cmpmbct. The results coincide with those in~\cmpwbct,
%thus is not shown.

%\normalsize
%The performance of ~\cmpmpct is similar to the performance of~\cmpwpct, and is not shown.

%\vspace{-.8ex}
\stitle{Summary.} %From the experiments
We can summarize the results as follows.
(a) Our inference algorithms can infer cascades effectively.
For example, the original cascades and the ones inferred by~\cmpwpct
have structural similarity (measured by $\precision_e$) of higher than $75\%$ in both real-life datasets.
%For example, even when $85\%$ of the nodes are unknown, \cmpwpct
%infers more than $50\%$ of the unknown nodes.
(b) Our algorithms scale well with
the sizes of the cascades, and uncertainty. %In particular,
They seldom demonstrated their worst-case complexity. %in our experiments.
For example, even for cascades with $240$ nodes,
all of our algorithms take less than two seconds.

\vspace{-1ex}
\section{Conclusion}

\vspace{-1ex}
\begin{table}[tb!]
\vspace{-3ex}
%\vspace{-4ex}
%\begin{small}
\scriptsize
\begin{center}
\begin{tabular}{|c|c|c|c|c|}
    \hline
  Problem   & Complexity & Approximation & time\\ % & Algorithm complexity \\
  \hline \hline
  %\bct &  bounded tree & \PTIME & - & $O(|E|)$ \\%&    \\
  %\hline % \cline{1-5}
  \mbct  & \NP-c, \APX-hard & $|X|$ & $O(|E|)$    \\
    \hline %\cline{1-5}
   \mwbct  & \NP-c,\APX-hard & $|X|*\frac{\log f_{max}}{\log f_{min}}$ & $O(|E|))$ \\
    \hline %\cline{1-5}
  % \pct  & perfect tree & \NP-complete & - & - \\

   \mpct (\sp tree) &\NP-c, \APX-hard  & $d$ &  $O|V^3|$ \\
   \hline
   \mwpct (\sp tree) & \NP-c, \APX-hard  & $d*\frac{\log f_{max}}{\log f_{min}}$ & $O|V^3|$ \\
   \hline
   \hline
 \mpct   & \NP-c, \APX-hard & -- & $O(|t_{max}*|V|^3)$ \\
    \hline %\cline{1-5}
   \mwpct  & \NP-c, \APX-hard & -- & $O(|t_{max}*|V|^3)$ \\
   \hline
  % \pct & \sp tree & \PTIME & - & $O(|E|)$ \\
  %  \hline
 % \mpct & \sp tree & \NP-complete & \APX-hard & $d$-approximable, $O(|V^3|)$ \\
%    \hline
%    \mwpct & \sp tree & \NP-complete & \APX-hard & $d*\frac{\log f_{max}}{\log f_{min}}$-approximable, $O(|V|^3)$ \\
    \hline
\end{tabular}
\end{center}
%\end{small}
\vspace{-3ex}
\caption{Summary: complexity and approximability}
\label{tab-results} %\vspace{-2ex}
\vspace{-6ex}
\end{table}
\normalsize

In this paper, we investigated cascade inference problem
based on partial observation.
We proposed the notions of consistent trees for capturing
the inferred cascades, namely, bounded consistent trees and
perfect consistent trees, as well as quantitative metrics by
minimizing either the size of the inferred structure or
maximizing the overall likelihood. We have established
the intractability and the hardness results for
the optimization problems as summarized in Table~\ref{tab-results}.
Despite the hardness, we developed approximation and heuristic algorithms for these problems,
with performance guarantees on inference quality,
We verified the effectiveness and efficiency of our
techniques using real life and synthetic cascades.
Our experimental results have shown that our methods are
able to efficiently and effectively infer the structure of information cascades.
%This work is the first step towards inferring cascade structures based on
%partial observation.
%We are extending our techniques
%over more application areas and to infer cascades following
%other cascade models.

%We also plan to improve our algorithms by leveraging data properties, filtering and indexing techniques. Another topic is to infer cascades which propagate across distributed networks.

%\vspace{4ex}
\renewcommand{\baselinestretch}{0.9}
\bibliographystyle{abbrv}
\begin{small}
%\setlength{\itemsep}{-5mm}
%\vspace{-2ex}
\bibliography{ref}
\end{small}

\end{document}